\documentclass[twocolumn,10pt]{IEEEtran}
\usepackage{braket}

\input{def.tex}
\usepackage{dsfont}
\usepackage{minibox}
\DeclareSymbolFont{matha}{OML}{txmi}{m}{it}
\DeclareMathSymbol{\varv}{\mathord}{matha}{118}
\usepackage{eurosym}
\usepackage{hhline,physics}

\usepackage{multicol}
\usepackage{algorithm}
\usepackage{graphicx}
\usepackage{textcomp}
\usepackage{caption,pifont}
\usepackage[multiple]{footmisc}

\usepackage{algorithmic}

\newcommand{\trans}{^{\mathsf{T}}}
\newcommand{\herm}{^{\H}}
\makeatletter
\renewcommand{\baselinestretch}{0.98}

\IEEEoverridecommandlockouts
\begin{document}
	\title{STAR-RIS Assisted Cell-Free Massive MIMO System Under Spatially-Correlated Channels}
	\author{Anastasios Papazafeiropoulos, Hien Quoc Ngo, Pandelis Kourtessis, Symeon Chatzinotas
		\thanks{Copyright (c) 2015 IEEE. Personal use of this material is permitted. However, permission to use this material for any other purposes must be obtained from the IEEE by sending a request to pubs-permissions@ieee.org.}
		 \thanks{A. Papazafeiropoulos is with the Communications and Intelligent Systems Research Group, University of Hertfordshire, Hatfield AL10 9AB, U. K., and with SnT at the University of Luxembourg, Luxembourg. Hien Quoc Ngo is with the School of Electronics, Electrical Engineering
			and Computer Science, Queen’s University Belfast, Belfast BT7 1NN, U.K. P. Kourtessis is with the Communications and Intelligent Systems Research Group, University of Hertfordshire, Hatfield AL10 9AB, U. K. S. Chatzinotas is with the SnT at the University of Luxembourg, Luxembourg. A. Papazafeiropoulos was supported  by the University of Hertfordshire's 5-year Vice Chancellor's Research Fellowship. The work of H. Q. Ngo was supported by the U.K. Research and Innovation Future Leaders Fellowships under Grant MR/X010635/1. S. Chatzinotas   was supported by the National Research Fund, Luxembourg, under the project RISOTTI. E-mails: tapapazaf@gmail.com, hien.ngo@qub.ac.uk p.kourtessis@herts.ac.uk, symeon.chatzinotas@uni.lu.}}
	\maketitle\vspace{-1.7cm}
	\begin{abstract}
		\textcolor{black}{This paper investigates the performance of  downlink simultaneous transmitting and reflecting reconfigurable
		intelligent surface (STAR-RIS)-assisted cell-free (CF) massive multiple-input multiple-output (mMIMO) systems, where user equipments (UEs) are located on both sides of the RIS}.
		We account for correlated Rayleigh fading and multiple antennas per access point (AP), while the maximum ratio (MR) beamforming is applied for the design of the active beamforming in terms of  instantaneous channel state information (CSI). Firstly, we rely on an aggregated channel estimation approach that reduces 
		the overhead required for channel estimation while providing sufficient information for data processing. We  obtain the normalized
		mean square error (NMSE) of the channel estimate per AP, and design the passive beamforming (PB) of the surface based on the long-time statistical CSI. Next, we derive the received signal in the asymptotic regime of numbers of APs and surface elements. Then, we obtain a closed-form expression of the downlink achievable rate for arbitrary numbers of APs  and STAR-RIS elements under statistical CSI.   Finally, based on the derived expressions, the numerical results show the feasibility and the advantages of deploying a STAR-RIS into conventional CF mMIMO systems. In particular, we theoretically analyze the properties of STAR-RIS-assisted CF mMIMO systems and reveal explicit insights in terms of the impact of channel correlation, the number of surface elements, and the pilot contamination on the achievable rate.
	\end{abstract}
	
	\section{Introduction}
	Massive multiple-input multiple-output (mMIMO) has already been identified as one of the main advancements in fifth-generation (5G)  networks with remarkable improvements in the rate, latency, reliability,  and coverage \cite{Boccardi2014}. MMIMO  can be deployed in not only collocated but also in distributed setups. The advantage of the collocated mMIMO setup is the low backhaul requirements. On the other hand, the benefit of a distributed mMIMO architecture, where the transmit antennas are scattered across a large area, is the extended coverage by exploiting diversity against shadow fading.  Unfortunately, both architectures are cellular-based since each base station (BS) serves the user equipments (UEs) within the cell boundaries, i.e., the  inter-cell interference is a significant limiting factor. 
	
	In this direction,  cell-free (CF) mMIMO has appeared as a promising solution to mitigate the effect of inter-cell interference with the combination of centralised processing and  CF design \cite{Ngo2017,Nayebi2017,Ngo2018,Buzzi2019}.  Basically, CF mMIMO relies on a distributed mMIMO layout, where a large number of access points allocated in a given area serve a smaller number of UEs. All access points (APs) cooperate to \textcolor{black}{serve} the UEs in the same  frequency \textcolor{black}{resources} by alleviating any cell boundaries. CF mMIMO presents an enhanced performance by reaping the benefits of network MIMO and distributed mMIMO together within the close  proximity of UEs to the APs. Note that even maximum ratio (MR) beamforming and match filtering result in high spectral efficiency (SE) in the downlink and uplink, respectively. Hence, CF  has been at the centre of  attention in the last few years. However, a large-scale deployment of APs can result in high costs regarding hardware and power since even single-antenna  APs are accompanied by a radio frequency (RF) chain. Also, in urban areas, uniform coverage is not attainable due to the uneven layout of tall buildings and the blocked areas due to obstacles. The simple solution of just deploying more APs to improve their coverage is not affordable since it leads to excessive deployment cost and energy consumption.
	
	Fortunately, reconfigurable intelligent surface (RIS) has emerged as a promising technology that can shape the propagation environment without any increase in the transmit power or the number of APs \cite{Wu2019,Basar2019,DiRenzo2020}. In particular, a RIS consists of an array of a large number of passive reflective elements  that induce phase shifts on the electromagnetic impinging waves towards improving the coverage, the SE, and the energy efficiency without any additional power \cite{DiRenzo2020,Wu2020}. Also, not only a RIS has a low deployment cost and can be installed easily, but its performance is comparable to mMIMO with a lower number of antennas and reduced transmit power  \cite{DiRenzo2020}. Thus, the integration of RIS is of great practical interest to decrease the equipment concerning the number of APs, the transmit power,  and improve the quality of service of blocked users.
	
	Although RIS has been suggested for  different communication scenarios due to its  numerous advantages \cite{Wu2019,Basar2019,DiRenzo2020,Kammoun2020,Yang2020b,Pan2020, Papazafeiropoulos2021,Papazafeiropoulos2021a,Papazafeiropoulos2022b}, most works have assumed that both the transmitter and receiver are located on the same side of  the  panel. In practise,  real-world applications require UEs located on both sides of the RIS. Thanks to recent advancements in programmable metamaterials, a new technology has appeared known as  simultaneous transmitting and reflecting RIS (STAR-RIS) that provides full space coverage by adjusting  the amplitudes and the phases of impinging waves \cite{Xu2021,Mu2021,Niu2021,Wu2021,Niu2022,Papazafeiropoulos2023}.  For example, in \cite{Mu2021}, it was shown that STAR-RIS has better performance than reflecting only RIS and three protocols were presented for optimising the panel performance namely energy splitting (ES), mode switching (MS), and time switching (TS). \textcolor{black}{Recently, in \cite{Papazafeiropoulos2023}, the achievable rate for  STAR-RIS assisted mMIMO channels was studied while assuming spatially correlated channels.}
	
	Many works have relied on the perfect channel state information (CSI) assumption, which is unrealistic for practical RIS-aided systems due to their lack of any  radio sources to receive or send any pilot signals \cite{Basar2019,DiRenzo2020}. On this topic, early papers on channel estimation (CE) protocols proposed ON/OFF strategies, which have large training overhead \cite{Mishra2019}. Other works with lower overheads are based on the ideas of RIS elements grouping or channel sparsity exploitation \cite{Yang2020b,He2019}. Notably, additionally to the training overhead the RIS optimization based on instantaneous CSI, changing at each coherence interval, increases the complexity and overhead.
	
	To address the burden of RIS optimization based on instantaneous CSI,  authors have turned their attention to designing  the RIS parameters by exploiting only statistical CSI in terms of large statistics such as correlation matrices and path-losses, which are slowly-varying (every several coherence intervals) \cite{Han2019,Zhao2020,Abrardo2021,Papazafeiropoulos2021b,Papazafeiropoulos2021,Zhi2022,Papazafeiropoulos2022a,Papazafeiropoulos2023}. Especially, in \cite{Zhao2020},  a maximisation of the achievable sum rate of a RIS-assisted multi-user multi-input single-output (MU-MISO) system took place by using a two-timescale transmission protocol, where precoding was designed in terms of instantaneous CSI, while the RIS   phase shifts were optimised by using statistical CSI. Moreover,  the two-timescale protocol was supplied to study the impact of hardware impairments on the sum rate and the minimum rate in \cite{Papazafeiropoulos2021} and \cite{Papazafeiropoulos2021b}, respectively. 
	
	Meanwhile most initial works on RIS  considered independent fading, but this is impractical as shown in \cite{Bjoernson2020}. \textcolor{black}{In particular, in \cite{Bjoernson2020}, it was proved that the Rayleigh fading channel model is not physically appearing when using an RIS in an isotropic  scattering environment, thus it should not be used.} Hence, many recent works have taken RIS correlation into consideration for conventional RIS-assisted systems but this consideration  for STAR-RIS systems is limited. Also,  the majority of works on STAR-RIS  have only considered a single UE on each side of the surface, e.g., see \cite{Mu2021}.
	
	Despite the necessity for introducing RIS into CF mMIMO systems,  the current literature is still in its infancy \cite{Zhang2021a,Siddiqi2022,Liu2021,Shi2022,VanChien2022,Shi2022a,Zhang2021b,Ge2022,Zhang2022,Yang2021}. Most works have  formulated optimization problems for RIS-assisted CF mMIMO with certain communication objectives but with the impractical assumptions of perfect CSI. In \cite{Zhang2021a}, a hybrid
	beamforming  scheme was proposed to decompose the original optimization problem into the digital beamforming subproblem and the RIS-based analog beamforming
	subproblem. In \cite{Siddiqi2022}, an  alternate optimization (AO) algorithm was proposed for the solution of these two beamforming problems  by using the sequential programming (SP) method and zero-forcing (ZF) beamforming. In \cite{Liu2021}, the maximization of the energy efficiency of the worst user in a wideband RIS-aided cell-free network took place by an iterative precoding algorithm
	using Lagrangian transform and fractional programming. The uplink SE was considered in \cite{Shi2022}  for spatially correlated RIS-aided cell-free mMIMO systems in the case of Rician fading channels. Moreover, in \cite{VanChien2022}, the aggregated channel, including both the direct and indirect links, was considered. In \cite{Shi2022a}, it was shown that generalized maximum ratio combining could double the achievable date rate over  the maximum ratio combining. The joint
	precoding design at BSs and RISs took place in  \cite{Zhang2021b} in the case of a wideband RIS-aided cell-free network to maximize the network capacity. In \cite{Ge2022}, a generalized  superimposed channel estimation scheme was proposed for an uplink cell-free  mMIMO system to enhance the  wireless coverage  and SE. In \cite{Zhang2022} and \cite{Yang2021},  the secure communication in a RIS-aided cell-free mMIMO system under the presence of  active eavesdropping and   the beamforming optimization for RIS-aided simultaneous wireless  information and power transfer  in cell-free mMIMO networks were investigated, respectively. In \cite{Dai2022},  the two-time scale protocol was considered for optimizing  the achievable uplink rate.
	
	
	\subsection{Motivation and contributions}
	Different from \cite{Zhang2021a,Siddiqi2022,Liu2021,Shi2022,VanChien2022,Shi2022a,Zhang2021b,Ge2022,Zhang2022,Yang2021}, we develop an analytical framework for the study of the downlink of a STAR-RIS-assisted CF mMIMO system by using statistical CSI,  where the aggregate APs-UEs channels, based on imperfect CSI, are considered  with MR beamforming for information decoding. Also, contrary to \cite{VanChien2022}, we account for APs with multiple antennas, while compared to \cite{Dai2022}, we consider correlated Rayleigh fading. \textcolor{black}{Although in \cite{Papazafeiropoulos2023}, we have also assumed correlated fading in STAR-RIS-assisted systems, we have focused on colocated mMIMO systems, while, in this work, we study CF mMIMO systems, which is another network architecture based on a distributed structure  with special characteristics such as enhanced  macro-diversity gain.}
	
	Thus, we avoid the estimation of all individual channels and STAR-RIS optimization at each channel realisation.  Specifically, firstly, we  describe the channel estimation, where each AP estimates the cascaded channel by the linear minimum mean square error (LMMSE)  estimation technique. We introduce the normalized
	mean square error (NMSE) of the channel estimate, and we design the passive beamforming of RISs based on the long-time statistical CSI, where the STAR-RIS parameters, i.e., the amplitudes   and phase shifts are optimised simultaneously. Next, the APs apply MR beamforming  to send their data. A closed-form expression of the downlink sum-SE is obtained in terms of statistical CSI. 
	
	Our main contributions are summarised as follows:
	\begin{itemize}
		\item We consider a STAR-RIS-assisted CF mMIMO
		under spatially correlated channels. All APs have multiple antennas and obtain the LMMSE of the instantaneous cascaded channel in the uplink pilot training phase for each AP-UE link with lower overhead compared to other channel estimation methods such as the ON/OFF strategy. Notably, we account for  pilot contamination based on an arbitrary  	pilot reuse pattern. This is the only work that has introduced STAR-RIS on CF mMIMO systems as far as the authors are aware.
		\item We introduce the NMSE of the channel estimate for each AP. Contrary to \cite{VanChien2022}, our work focuses on STAR-RIS and assumes multiple antennas per AP, which requires a different analysis with more complicate manipulations. Based on the NMSE, we optimize the \textcolor{black}{passive beamforming (PB)} of the STAR-RIS, which includes a simultaneous optimization of the  amplitudes   and phase shifts of the surface.
		\item We show the performance of the system in the large number of APs and surface elements regime, where the small-scale fading,  the additive noise, and the coherent interference average out even with MR beamforming,  while the  received signal includes only the desired  	signal and the coherent interference. 
		\item We  derive a closed-form expression for the  downlink achievable sum SE in terms of only statistical CSI. The analytical expression reveals the impact of spatial correlation, pilot contamination, numbers of APs and elements, and phase shifts of the surface.
		\item Simulation results verify the analytical expressions and present the superiority of STAR-RIS against conventional RIS-assisted CF mMIMO systems.
	\end{itemize}  
	
	\textit{Paper Outline}: The remainder of this paper is organized as follows. Section~\ref{System} introduces the system model of a STAR-RIS-assisted CF mMIMO system with correlated Rayleigh fading. Section~\ref{ChannelEstimation} presents the channel estimation of the cascaded channel the design of the passive beamforming  based on the long-time statistical CSI. Section~\ref{PerformanceAnalysis} presents the downlink data transmission with the derived downlink sum SE. 
	The numerical results are provided in Section~\ref{Numerical}, and Section~\ref{Conclusion} concludes the paper.
	
	\textit{Notation}: Vectors and matrices are represented by boldface lower and upper case symbols, respectively. The notations $(\cdot)^\T$, $(\cdot)^\H$, and $\tr\!\left( {\cdot} \right)$ denote the transpose, Hermitian transpose, and trace operators, respectively. Moreover, the notations  $\EE\left[\cdot\right]$ and $ \nabla(\cdot) $ describe  the expectation  operator and gradient operator, respectively. The notation  $\diag\left(\bA\right) $ denotes a vector with elements equal to the  diagonal elements of $ \bA $, the notation  $\diag\left(\bx\right) $ describes a diagonal  matrix whose elements are $ \bx $. \textcolor{black}{$ O(\cdot) $ denotes  the limiting behavior of a function when the argument tends towards  infinity,  $\log()$ denotes the natural logarithm and 
		$\left\lceil \cdot \right\rceil$ 
		denotes the smallest integer that is larger than or equal to the argument } Also,   $\bb \sim \cC\cN{(\b0,\mathbf{\Sigma})}$ denotes a circularly symmetric complex Gaussian vector with zero mean and a  covariance matrix $\mathbf{\Sigma}$. 
	
	\section{System Model}\label{System}
	In this section, we present a CF mMIMO system assisted by a STAR-RIS. For the sake of demonstrating a practical scenario, we assume that the CF mMIMO system is implemented outdoor while communication takes place with both indoor and outdoor UEs  through a STAR-RIS deployed on the wall of the building hosting the indoor UEs as shown  in Fig~\ref{Fig1}. Also, all UEs are assumed far from the APs or behind obstacles to justify the implementation of a RIS.
	
	In particular, we consider $ M $ multi-antenna APs with $ L $ antennas each, which are  connected to a central processing unit
	(CPU) and serve simultaneously $ K $ \textcolor{black}{edge} UEs in total on the same time and frequency resource.\footnote{\textcolor{black}{It would also be very interesting to consider non-edge users, where the APs can be located at both sides of the surface but this is left for future work.}}  
The group of UEs consists of  $ \mathcal{K}_{t}=\{1,\ldots,K_{t} \} $ UEs  located in the transmission region $ (t) $ and $ \mathcal{K}_{r}=\{1,\ldots,K_{r} \} $ UEs located in the reflection region $ (r) $, respectively. In other words, $ K=K_{t}+K_{r} $. For the sake of a better representation of the analysis below, we  denote as $ \mathcal{W} = \{w_{1}, w_{2}, ..., w_{K}\}  $  a set that defines the RIS operation mode for each of $ K $ UEs in a unified way. Hence,  if the $ k $th UE is located behind the STAR-RIS ($ k\in   \mathcal{K}_{t}$), then $ w_{k} = t $, while $ w_{k} = r $, if  the $ k $th UE is found in front of  the STAR-RIS ($ k\in   \mathcal{K}_{r}$), i.e.,  the  APs and the $ k $th UE are found at the same side. We assume  no direct links exist between the APs and UEs in both regions due to obstacles.  
	
	\begin{figure}[!h]
		\begin{center}
			\includegraphics[width=0.7\linewidth]{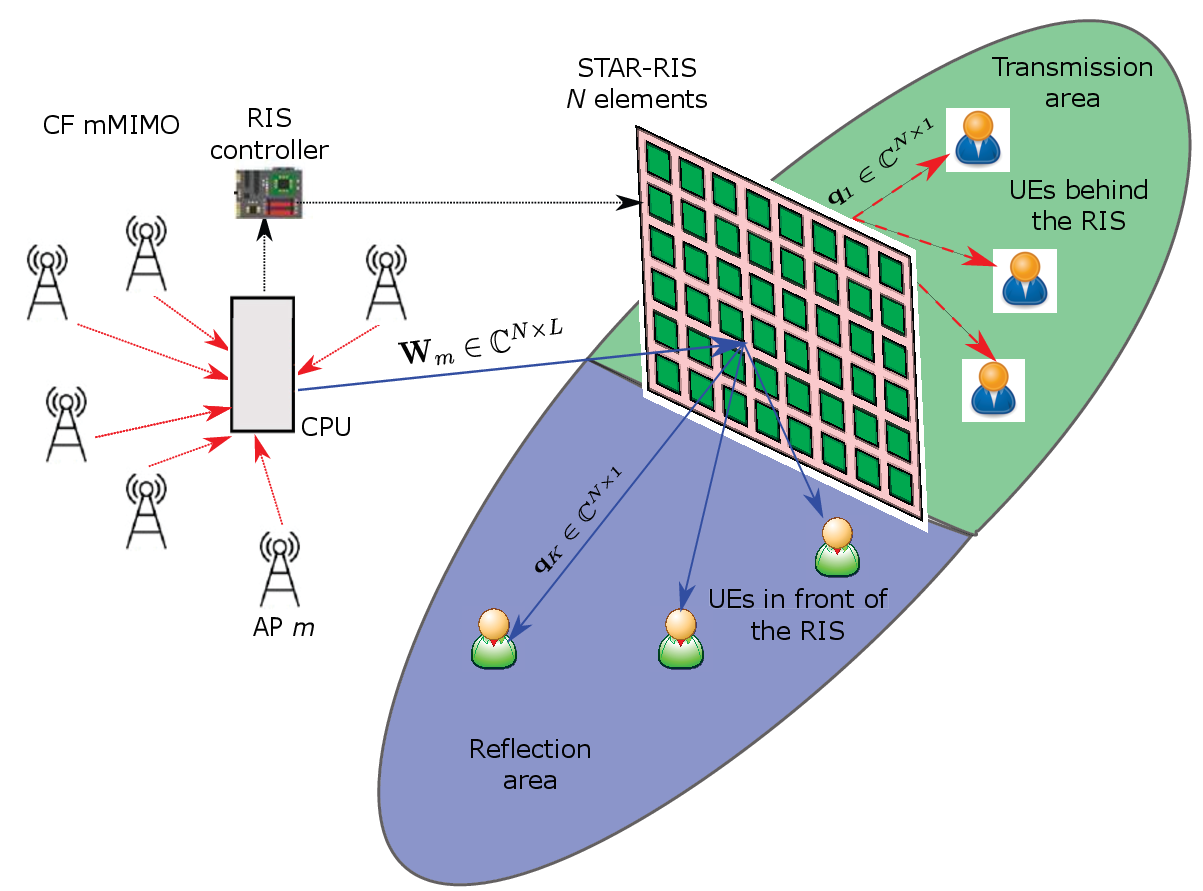}
			\caption{{ A CF mMIMO STAR-RIS-assisted system with multiple antennas per AP and with multiple UEs at transmission and reflection regions.  }}
			\label{Fig1}
		\end{center}
	\end{figure} 
	
	The STAR-RIS, being one of the focal points of the system model,  consists of a uniform planar array (UPA) that can modify the amplitudes and phases of impinging signals. Each row of the array includes $ N_{\mathrm{h}} $  elements and  each column consists of $ N_{\mathrm{v}} $ elements, i.e.,  the total number of RIS elements is $ N= N_{\mathrm{h}} \times N_{\mathrm{v}} $ and defines a set of $\mathcal{N}=\{1, \ldots, N \}$ elements. The ability of the STAR-RIS is characterized by the simultaneous configuration of  the transmitted ($ t $) and reflected ($ r $) signals by two independent coefficients.  Especially, the transmitted 	and reflected signal by the $ n $th RIS element are described by   $ t_{n} =( {\beta_{n}^{t}}e^{j \phi_{n}^{t}})s_{n}$ and $ r_{n}=( {\beta_{n}^{r}}e^{j \phi_{n}^{r}})s_{n} $, respectively, where $ s_{n} $ denotes the impinging signal.  We assume that the amplitude and phase parameters  $ {\beta_{n}^{k}}\in [0,1] $ and $ \phi_{n}^{k} \in [0,2\pi)$, where the $  k $th UE can be in any of the two regions that  corresponds also to the  RIS mode, i.e.  transmission ($ t $) or reflection  ($ r $)
	\cite{Xu2021}, are independent.\footnote{\textcolor{black}{In practice, the phases of the reflecting and transmitting coefficient are coupled with each other.  However, this consideration requires a separate analysis regarding the optimization, and is left for future work \cite{Xu2022}.}} Although  $ \phi_{n}^{t} $ and $ \phi_{n}^{r} $ can be chosen independently,  the choice of the amplitudes is based on the correlation expressed by the law of energy conservation as
	\begin{align}
		(\beta_{n}^{t})^{2}+(\beta_{n}^{r})^{2}=1,  \forall n \in \mathcal{N}.
	\end{align}
	
	Both the amplitudes and the phase shifts are adjusted by a controller through a backhaul link that enables the exchange of information between the APs and the STAR-RIS. Below, we denote $ \theta_{k}^{w_{k}}=e^{j\phi_{k}^{w_{k}}} $ for the sake of better exposition.
	
	\subsection{STAR-RIS model Protocols}
	A STAR-RIS generally operates in the ES mode, while the MS mode can also  be used, which has lower performance but also lower complexity. Another option is the TS protocol, but herein, we will focus on the  first two protocols due to limited space, while the study of the   TS protocol is left for future work.
	
	\subsubsection{ES protocol} All RIS elements serve simultaneously all UEs despite their locations. In other words, UEs can be found in both   $ t $ and $ r $ regions. The  PB for UE $ k
	$  is expressed as $ \bPhi_{w_{k}}^{\mathrm{ES}}=\diag( {\beta_{1}^{w_{k}}}\theta_{1}^{w_{k}}, \ldots,  {\beta_{N}^{w_{k}}}\theta_{N}^{w_{k}}) \in \mathbb{C}^{N\times N}$, where $ \beta_{n}^{w_{k}} \ge 0 $, $ 		(\beta_{n}^{t})^{2}+(\beta_{n}^{r})^{2}=1 $, and $ |\theta_{i}^{w_{k}}|=1, \forall n \in \mathcal{N} $.

	\subsubsection{MS protocol} This protocol suggests the  partitioning of the  RIS elements   into two groups of $ N_{t} $ and $ N_{r} $ elements that serve UEs in the  $  t$ and $ r $ regions, respectively. In other words,  $ N_{t}+N_{r}=N $. The PB for $  k \in \mathcal{K}_{t} $ or $  k \in \mathcal{K}_{r} $
	is given by $ \bPhi_{w_{k}}^{\mathrm{MS}}=\diag( {\beta_{1}^{w_{k}}} \theta_{1}^{w_{k}}, \ldots,  {\beta_{N}^{w_{k}}}\theta_{N}^{w_{k}}) \in \mathbb{C}^{N\times N}$, where $ \beta_{n}^{w_{k}}\in \{0,1\}$, $ 	(\beta_{n}^{t})^{2}+(\beta_{n}^{r})^{2}=1 $, and $ |\theta_{i}^{w_{k}}|=1, \forall n \in \mathcal{N} $. In other words,  the amplitude coefficients for transmission and reflection are restricted to binary values, which means that the MS protocol is a special case of the ES protocol, i.e., it is inferior.
	
	\subsection{Channel Model}\label{ChannelModel} 
	We assume  narrowband quasi-static block fading channels based on the standard time-division-duplex (TDD) protocol, which is used in CF mMIMO systems and enables channel reciprocity. Specifically, each block has a duration of $\tau_{\mathrm{c}}$ channel uses, while $\tau$ channel uses are allocated for the uplink training phase and $\tau_{\mathrm{c}}-\tau$ channel uses are allocated for the downlink data transmission phase.
	
	Let  $ \bw_{ml} \in \mathbb{C}^{N \times 1}$, and  $ \bq_{k} \in \mathbb{C}^{N \times 1}$ be  the channel between the $ l $th antenna of AP $m  $ and the RIS, and the channel between the RIS and UE $ k $, respectively. Note that $ 	\bW_{m}=[\bw_{m1}, \ldots,\bw_{mL}] \in \mathbb{C}^{N \times  L}$ is the concatenated downlink channel between AP $ m $ and the RIS. Especially, the links including the RIS  underlie unavoidable correlated Rayleigh fading \cite{Bjoernson2020}.
	Thus, we have
	\begin{align}
		\bW_{m}&=\sqrt{\tilde{ \beta}_{m}}\textcolor{black}{\bR_{\mathrm{RIS},m}}^{1/2}\bD_{m} \textcolor{black}{\bR_{\mathrm{AP},m}} ^{1/2},\label{Wchannel}\\
		\bq_{k}&=\sqrt{\tilde{ \beta}_{k}}\textcolor{black}{\bR_{\mathrm{RIS},k}}^{1/2}\bc_{k},	\label{Qchannel}
	\end{align}
	where  $ \bR_{\mathrm{RIS},m},  \in \mathbb{C}^{N \times N} $,  $ \bR_{\mathrm{RIS},k},  \in \mathbb{C}^{N \times N} $, and $ {\bR}_{\mathrm{AP}} \in \mathbb{C}^{L \times L} $ express the deterministic Hermitian-symmetric positive semi-definite correlation matrices at the RIS and the AP $ m $, which are assumed to be known by the network. Also, $ \mathrm{vec}(\bD_{m})\sim \mathcal{CN}\left(\b0,\Id_{LN}\right) $ and $ \bc_{k} \sim \mathcal{CN}\left(\b0,\Id_{N}\right) $ express the corresponding fast-fading vectors. Note that, in the case of isotropic
	scattering with uniformly distributed multipath components, $ \bR_{\mathrm{RIS},m} $ and $ \bR_{\mathrm{RIS},k} $ are modeled as in \cite{Bjoernson2020}. In particular, we have $ [\bR]_{i,j}=d_{\mathrm{H}}d_{\mathrm{V}}\mathrm{sinc}(2\|\bu_{i}-\bu_{j}\|/\lambda) $, where $ \lambda $ is the wavelength, $\mathrm{sinc}(x)=\sin(\pi x)/(\pi x)  $ is the sinc
	function, $ d_{\mathrm{H}}d_{\mathrm{V}} $ is the  size of each element of the RIS with  $ d_{\mathrm{H}} $ and $ d_{\mathrm{V}} $ being  the horizontal width and 
	the vertical height of each RIS element. Also, the location of the $ n $th element
	with respect to the origin is given by $ \bu_{n} =[0,\mod(n-1)N_{\mathrm{h}}d_{\mathrm{H}}, \lfloor(n-1)N_{\mathrm{h}}\rfloor d_{\mathrm{V}}]^{\T}$.
	Moreover, $\tilde{ \beta}_{m} $ and $\tilde{ \beta}_{k} $  express the path-losses  of the $ m $th AP-RIS and RIS-UE $ k $ links, respectively.  Note that the channel model in \eqref{Wchannel} is different from most related works, which assume that $ \bW_{m} $ is deterministic for the sake of analytical tractability \cite{Nadeem2020,Papazafeiropoulos2021}.$  $

	Given the PB,  the aggregated channel between  AP $ m $ and UE $ k $, consisted of the   cascaded channel,   is 
	\begin{align}
		\bh_{mk}= \bW_{m}^{\H}\bPhi_{w_{k}} \bq_{k},\label{cascaded}
	\end{align}
	and has a variance  $ \bR_{mk}=\EE\{\bh_{mk}\bh_{mk}^{\H}\} $ given by
	\begin{align}
		\bR_{mk}=\tr(\tilde{\bR}_{mk})\bR_{\mathrm{AP},m},\label{cov1}
	\end{align}
	where  $ \tilde{\bR}_{mk}=\tilde{\beta}_{mk}\bR_{\mathrm{RIS},m} \bPhi_{w_{k}} \bR_{\mathrm{RIS},k}  \bPhi_{w_{k}}^{\H} $. We have used the independence between $ \bW_{m} $ and $ \bq_{k} $, $\tilde{\beta}_{mk}= \tilde{ \beta}_{m}\tilde{ \beta}_{k} $, $ \EE\{	\bq_{k}	\bq_{k}^{\H}\} =\tilde{ \beta}_{k} \bR_{\mathrm{RIS},k}$,   $ \EE\{\bV \bB\bV^{\H}\} =\tr (\bB) \Id_{M}$ with $\bB  $ being a deterministic square matrix, and $ \bV $ being any matrix with independent and identically distributed (i.i.d.) entries of zero mean and unit variance. In the case that  $\bR_{\mathrm{RIS},m} =\bR_{\mathrm{RIS},k}=\Id_{N} $, $ \bR_{mk} $ does not depend on the phase shifts but only on the amplitudes, as also mentioned in \cite{Papazafeiropoulos2022}. For convenience, we provide below the lemma in \cite[Lem. 1]{Bjornson2015}.
	\begin{lemma}\label{lem0proof1}
		Let $ \bA $ be deterministic matrices, and $ \bx_{mk} $ is a vector of i.i.d. $ \mathcal{CN}(0,1) $ elements. Then, we have
		\begin{align}
			\EE\{|\bx_{mk}^{\H}\bA\bx_{mk}|^{2}\}=\tr(\bA\bA^{\H})+(\tr(\bA))^{2}.\label{exp1}
		\end{align}
	\end{lemma}
	
	Also, regarding the cascaded channel given by \eqref{cascaded}, we observe that it is the product of weighted complex Gaussian and spatially correlated random variables. Despite the statistical complexity of \eqref{cascaded}, the following lemma provides useful expressions with respect to the expectation of the channel, which are required in Sec. \ref{PerformanceAnalysis}.C.
	
	\begin{lemma}\label{lem0proof}
		Let $ \bA $ and $ \bB $ be deterministic matrices, then we have \eqref{exp2}-\eqref{exp6} at the top of the next page.
		\begin{figure*}
			\begin{align}
				&\EE\{\bh_{mk}^{\H}\bA\bh_{mk}\}=\tr(\bA \bR_{mk})\label{exp2},\\
				&\EE\{ \bh_{mk}^{\H}\bA\bh_{mk}\bh_{nk}^{\H}\bB\bh_{nk}\}\!=\!\tr(\bR_{\mathrm{AP},m}\bA)\!\tr(\bR_{\mathrm{AP},m}\bB)\!\big(\!\tr(\tilde{\bR}_{mk})\tr(\tilde{\bR}_{nk})\!+\!\tr(\tilde{\bR}_{mk}\tilde{\bR}_{nk})\!\big), m\ne n\label{exp3}\\
				&	\EE\{ \bh_{mk}^{\H}\bA\bh_{ml}\bh_{nl}^{\H}\bB\bh_{nk}\}=\tr(\bR_{\mathrm{AP},m}\bA)\tr(\bR_{\mathrm{AP},m}\bB)	\tr(\tilde{\bR}_{mk}\tilde{\bR}_{nl}), m\ne n\label{exp4}\\
				&	\EE\{| \bh_{mk}^{\H}\bA\bh_{mk}|^{2}\}\approx \tr(\bR_{\mathrm{AP},m}\bA\bR_{\mathrm{AP},m}\bA^{\H})\Big(\!\!\tr\big(\tilde{\bR}_{mk}^{2}\big)+\big(\!\tr\big(\tilde{\bR}_{mk}\big)\!\big)^{2}\!\Big),\label{exp5}\\
				&	\EE\{| \bh_{mk}^{\H}\bA\bh_{ml}|^{2}\}\approx\tr(\bR_{\mathrm{AP},m}\bA\bR_{\mathrm{AP},m}\bA^{\H})\tr(\tilde{\bR}_{ml})\tr(\tilde{\bR}_{mk}),k\ne l\label{exp6}.
			\end{align}
			\hrulefill
		\end{figure*}
	\end{lemma}
	\begin{proof}
		See Appendix~\ref{lem0}.	
	\end{proof}

\section{Channel Estimation}\label{ChannelEstimation}
A STAR-RIS is basically implemented by nearly passive elements without any RF chains. Thus, it  cannot process the estimated channels, i.e., it cannot obtain the received pilots by UEs, and it  cannot transmit any pilot sequences to the APs for channel estimation.  Given that perfect CSI  is unavailable and based  on the TDD protocol, we resort to a channel estimation method obtaining the estimation of the cascaded channel by an uplink training phase with pilot symbols \cite{Bjoernson2017}, while being indifferent to the individual channels.\footnote{\textcolor{black}{There are two general approaches for channel estimation. One focuses on the estimation of the individual channels such as [42],  and the other one obtains the estimated aggregated channel [13], [28].  In particular, in this work, we employ the second approach, which has lower overhead and provides the estimated channel in closed form. This method has not been used before for STAR-RIS CF mMIMO. This approach allows to highlight one of the advantages of this work, which is making   the expression for the channel estimation "looking" the same for both types of users belonging to different areas of the RIS. However, the channel estimation   is different since   $  {\bR}_{mk} $ including the expressions corresponding to the   phases shifts is different, for users in  $ t $ and  $ r $ regions. In other words, we have managed in a smart way to introduce the standard channel estimation for  multiple-user SIMO to STAR-RIS CF mMIMO, which has not taken place before. Also, the treatment of imperfect CSI via LMMSE seems to be standard but it is not since it concerns the aggregated channel vector.  Note that other papers have not employed this efficient way for channel estimation on STAR-RIS CF mMIMO.}}

\subsection{LMMSE Estimation}
We denote as $ \mathcal{P}_{k} $ the set of indices of UEs that share the same pilot sequence as UE $ k $. All UEs  in  $ t $ or $ r $ region that belong to $ \mathcal{P}_{k} $ share the same orthogonal   pilot sequences as UE $ k $. Specifically, we denote by $\bx_{k}=[x_{k,1}, \ldots, x_{k,\tau}]^{\T}\in \mathbb{C}^{\tau\times 1} $ the pilot sequence of UE $ k  $ that can be found in any of the two regions, where the duration of the uplink training phase is $ \tau $ symbols. Note that the pilot sequences are mutually orthogonal, i.e.,  $ \bx_{k}^{\H}\bx_{l}=0$, if  $l~\not\in \mathcal{P}_{k}$ and $ \bx_{k}^{\H}\bx_{l}= 1$, if  $l~ \in \mathcal{P}_{k}$.\footnote{\textcolor{black}{ Ideally, pilot sequences assigned to all users should be pairwisely orthogonal. This requires $ \tau >= K $. In many practical scenarios, the number of users is large or/and the coherence interval is short (i.e. in high mobility environments), the above condition cannot be fulfilled. Thus, orthogonal pilot sequences should be re-used among the users. This causes pilot contamination effect which may reduce the system performance significantly. To reduce the effect of pilot contamination, many pilot assignment/grouping schemes have been proposed in cell-free massive MIMO. For examples, in \cite{Ngo2017}, greedy pilot assignment is proposed, while \cite{Buzzi2020,Liu2019,Zeng2021}  proposed to new pilot assignment schemes based on Hungarian, Tabu-search, and weighted graphics algorithms, respectively. The pilot assignment for our considered systems is out scope of our work. It requires a comprehensive study and is left for future work.} } Moreover, all UEs use the same normalized signal-to-noise ratio (SNR) $ p $ for transmitting each pilot symbol during the training phase.

Let us assume UE $ k  $ transmitting the pilot sequence $ \sqrt{p\tau}\bx_{k} $. Then,  the received signal by the $ m $th AP is written as

\begin{align}
	\bY_{\mathrm{p},m} =\sum_{k=1}^{K}\sqrt{p \tau}\bW_{m}^{\H}\bPhi_{w_{k}} \bq_{k}\bx_{k}^{\H}+\bZ_{\mathrm{p},m},\label{estimated1}
\end{align}
where  $ \bZ_{\mathrm{p},m}$ is an $ L\times \tau $ noise matrix whose elements are independent and identically distributed
(i.i.d.)  $\mathcal{CN}\left(0,1\right)$ random variables. Next, the received signal in \eqref{estimated1} is projected to $ \bx_{k} $ to estimate the desired channel. Specifically, we obtain 
\begin{align}
	\by_{\mathrm{p},mk} &=\bY_{\mathrm{p},m}\bx_{k}\nn\\
	&=\sqrt{p \tau}\bW_{m}^{\H}\bPhi_{w_{k}} \bq_{k}\nn\\
	&+\sum_{l\in \mathcal{P}_{k} \backslash \{k\}}\sqrt{p \tau}\bW_{m}^{\H}\bPhi_{w_{l}} \bq_{l}+\bz_{\mathrm{p},mk},\label{up_rec}
\end{align}
where $ \bz_{\mathrm{p},mk} =\bZ_{\mathrm{p},m}\bx_{k}$ includes  i.i.d. $  \mathcal{CN}\left(0,1\right)$ random variables.

Having fixed the phase shifts, we apply the LMMSE method for estimating  $ \bh_{mk} $ at the AP even in the presence of STAR-RIS, where the cascaded channel is given by the product of  weighted complex Gaussian random variables.  Despite the complicate form of the cascaded channel, the following lemma provides a closed-form expression of the estimated channel.

\begin{lemma}\label{PropositionDirectChannel}
	The LMMSE estimate of the cascaded channel \textcolor{black}{$ \bh_{mk} $} between the $ m $th AP and  UE $ k $  is given by
	\begin{align}
		\hat{\bh}_{mk}=\sqrt{p \tau}\bR_{mk}\bQ_{m} \by_{\mathrm{p},mk},\label{estim1}
	\end{align}
	where $ \bQ_{m}\!=\! \left({p \tau}\sum_{l\in \mathcal{P}_{k} }\bR_{ml}+\Id_{M}\right)^{\!-1}$, and $ \by_{\mathrm{p},mk}$ is the noisy channel given by \eqref{up_rec}. Hence, the estimated channel has zero mean and variance given by
	\begin{align}
		\bPsi_{mk}=p \tau \bR_{mk}\bQ_{k}\bR_{mk}.\label{estimatedVariance}
	\end{align}
	The estimated channel 	$ \hat{\bh}_{mk} $ is uncorrelated with the channel estimation error $ \bee_{mk}=\bh_{mk}-\hat{\bh}_{mk}  $, which has zero mean and variance $ \EE \{|\bee_{mk}|^{2}\} $ given by $ \bR_{mk}- \bPsi_{mk}$.
\end{lemma}
\begin{proof}
	See Appendix~\ref{lem1}.	
\end{proof}

With the PB fixed, according to Lemma~\ref{PropositionDirectChannel}, it is shown that the cascaded channel can be estimated without any further increase of the pilot training overhead since only $ \tau \ge K $ symbols are required in each coherence interval as in standard CF mMIMO systems. The closed-form expression of the channel estimate will be employed below to derive the  NMSE of the channel estimate of  UE $ k $ at  AP $  m $ and optimize the PB by minimizing the NMSE.

Given that channel estimation is crucial in CF mMIMO systems, herein, we focus on the design of the PB to increase the quality of the channel estimation. Specifically, we define
\begin{align}
	\mathrm{NMSE}_{mk}&=\frac{\tr(\EE[(\hat{\bh}_{mk}-{\bh}_{mk})(\hat{\bh}_{mk}-{\bh}_{mk})^{\H}])}{\tr\left(\EE[{\bh}_{mk}{\bh}_{mk}^{\H}]\right)}\\
	&=1-\frac{\tr(\bPsi_{mk})}{\tr(\bR_{mk})}.\label{nmse1}
\end{align}

The NMSE, lying in the range $ [0,1] $, is suitable for measuring the channel estimation quality per AP. If orthogonal pilot signals are employed for every UE and the pilot power tends to infinity, the NMSE tends to zero. \textcolor{black}{Also, if independent Rayleigh fading is assumed, the correlation matrices in \eqref{nmse1} do no depend on the phase shifts. Hence, the $ 	\mathrm{NMSE}_{mk}$ cannot be optimized.}

\subsection{PB Optimization}

It is essential  to optimize the  performance  of STAR-RIS-assisted systems with respect to the PB, which depends on both amplitudes and phase shifts. As in \cite{VanChien2022}, we propose to optimize the PB by minimizing the NMSE from all UEs and APs.
In this direction and based on 
infinite-resolution phase shifters, we formulate the optimization problem for minimizing the  $ 	\mathrm{NMSE}_{mk} $ of  CF mMIMO STAR-RIS-assisted systems accounting for correlated fading and imperfect CSI as	
\begin{equation}
	\begin{IEEEeqnarraybox}[][c]{rl}
		\min_{\thetv,\betv}&\quad\mathrm{NMSE}=\sum_{m=1}^{M}\sum_{k=1}^{K}	\mathrm{NMSE}_{mk}\\
		\mathrm{s.t}&\quad (\beta_{n}^{t})^{2}+(\beta_{n}^{r})^{2}=1,  \forall n \in \mathcal{N}\\
		&\quad\beta_{n}^{t}\ge 0, \beta_{n}^{r}\ge 0,~\forall n \in \mathcal{N}\\
		&\quad|\theta_{n}^{t}|=|\theta_{n}^{r}|=1, ~\forall n \in \mathcal{N}
	\end{IEEEeqnarraybox}\label{Maximization}\tag{$\mathcal{P}1$}
\end{equation}
where $\thetv=[(\thetv^{t})^{\T}, (\thetv^{r})^{\T}]^{\T} $ and $\betv=[(\betv^{t})^{\T}, (\betv^{r})^{\T}]^{\T} $.

We would like to emphasize that contrary to \cite{VanChien2022}, we consider multiple antennas per AP, and we are going to optimize not only the phase shifts but also the amplitudes. Of course, we differentiate from \cite{VanChien2022} because we refer to STAR-RIS instead of conventional RIS. The optimization approach and the result is different. Hence, the equal phase shift design, demonstrated in \cite{VanChien2022}, does not take place here, and we rely on the projected  gradient ascent method (PGAM) to obtain a locally-optimal solution of \eqref{Maximization}.  \textcolor{black}{The proposed optimization method of the phase
	shifts of the RIS, i.e., the minimization of the sum NMSE  focuses on  improving the channel
	estimation quality, which is a critical objective in mMIMO systems. In particular, an improvement in the accuracy of channel estimation leads to a significant  enhancement of  the downlink net throughput \cite{VanChien2022}. Note that an alternative solution would be to optimize the PB based on the maximization of the downlink SE provided below. However, this approach is challenging,  and page limitations have led us to defer the optimization of the SE in future work.}

\begin{remark}
	Under independent Rayleigh fading conditions,  $	\mathrm{NMSE}_{mk} $ is independent of  $ \thetv $. Hence, its optimization is performed in terms of $ \betv $.
\end{remark}

The  problem \eqref{Maximization} is non-convex while the amplitudes and the phase shifts for transmission and reflection are coupled. For the sake of exposition, the feasible set of \eqref{Maximization} is defined by the sets $ \Theta=\{\thetv\ |\ |\theta_{i}^{t}|=|\theta_{i}^{r}|=1,i=1,2,\ldots N\} $, and $ \mathcal{B}=\{\betv\ |\ (\beta_{i}^{t})^{2}+(\beta_{i}^{r})^{2}=1,\beta_{i}^{t}\geq0,\beta_{i}^{r}\geq0,i=1,2,\ldots N\} $. Since the sets $\Theta$ and $ \mathcal{B}$  the projection operators can be obtained in closed-form, we  apply the  PGAM \cite[Ch. 2]{Bertsekas1999} to optimize $\thetv$ and $\betv$. The proposed PGAM, which increases the objective from the current iterate $(\thetv^{n},\betv^{n})$   towards the gradient direction,  consists of the following iterations
\begin{subequations}\label{mainiteration}\begin{align}
		\thetv^{n+1}&=P_{\Theta}(\thetv^{n}+\mu_{n}\nabla_{\thetv}	\mathrm{NMSE}_{mk}(\thetv^{n},\betv^{n})),\label{step1} \\ \betv^{n+1}&=P_{\mathcal{B}}(\betv^{n}+{\mu}_{n}\nabla_{\betv}	\mathrm{NMSE}_{mk}(\thetv^{n},\betv^{n})),\label{step2} \end{align}
\end{subequations}
where the superscript expresses the iteration count, $\mu_n$ is the step size for both $\thetv$ and $\betv$ while  $P_{\Theta}(\cdot) $ and $ P_{\mathcal{B}}(\cdot) $ are the projections onto $ \Theta $ and $ \mathcal{B} $, respectively. 

Although  the ideal step size should be
inversely proportional to the Lipschitz constant of the corresponding gradient, it is difficult to find it  for the  problem above. Fortunately,   Armijo-Goldstein backtracking line search  allows finding the step size at each iteration. For this reason, we define a quadratic approximation of $	\mathrm{NMSE}_{mk}(\thetv,\betv)$ as
\begin{align}
	&	Q_{\mu}(\thetv, \betv;\bx,\by)=	\mathrm{NMSE}_{mk}(\thetv,\betv)\nn\\
	&+\langle	\nabla_{\thetv}	\mathrm{NMSE}_{mk}(\thetv,\betv),\bx-\thetv\rangle-\frac{1}{\mu}\|\bx-\thetv\|^{2}_{2}\nn\\
	&+\langle\nabla_{\betv}	\mathrm{NMSE}_{mk}(\thetv,\betv),\by-\betv\rangle-\frac{1}{\mu}\|\by-\betv\|^{2}_{2}.
\end{align}

The step size $  \mu_{n} $ in \eqref{mainiteration} can be obtained as $ \mu_{n} = L_{n}\kappa^{u_{n}} $, where we assume that $ L_n>0 $,  $ \kappa \in (0,1) $, and  $ u_{n} $ is the
smallest nonnegative integer satisfying
\begin{align}
	\mathrm{NMSE}_{mk}(\thetv^{n+1},\betv^{n+1})\geq	Q_{L_{n}\kappa^{u_{n}}}(\thetv^{n}, \betv^{n};\thetv^{n+1},\betv^{n+1}),
\end{align}
which can be performed by an iterative procedure. Note that the step size at iteration $n$ is used as the initial step size at iteration $n+1$. The proposed PGAM is summarized in Algorithm \ref{Algoa1}. 
\begin{algorithm}[th]
	\caption{Projected Gradient Ascent Algorithm for the RIS Design\label{Algoa1}}
	\begin{algorithmic}[1]
		\STATE Input: $\thetv^{0},\betv^{0},\mu_{1}>0$, $\kappa\in(0,1)$
		\STATE $n\gets1$
		\REPEAT
		\REPEAT \label{ls:start}
		\STATE $\thetv^{n+1}=P_{\Theta}(\thetv^{n}+\mu_{n}\nabla_{\thetv}	\mathrm{NMSE}_{mk}(\thetv^{n},\betv^{n}))$
		\STATE $\betv^{n+1}=P_{B}(\betv^{n}+\mu_{n}\nabla_{\betv}	\mathrm{NMSE}_{mk}(\thetv^{n},\betv^{n}))$
		\IF{ $	\mathrm{NMSE}_{mk}(\thetv^{n+1},\betv^{n+1})\leq Q_{\mu_{n}}(\thetv^{n},\betv^{n};\thetv^{n+1},\betv^{n+1})$}
		\STATE $\mu_{n}=\mu_{n}\kappa$
		\ENDIF
		\UNTIL{ $	\mathrm{NMSE}_{mk}(\thetv^{n+1},\betv^{n+1})>Q_{\mu_{n}}(\thetv^{n},\betv^{n};\thetv^{n+1},\betv^{n+1})$}\label{ls:end}
		\STATE $\mu_{n+1}\leftarrow\mu_{n}$
		\STATE $n\leftarrow n+1$
		\UNTIL{ convergence}
		\STATE Output: $\thetv^{n+1},\betv^{n+1}$
	\end{algorithmic}
\end{algorithm}

\begin{proposition}\label{LemmaGradients}
	The complex gradients $ \nabla_{\thetv}		\mathrm{NMSE}_{mk}(\thetv,\betv) $ and  $\nabla_{\betv}	\mathrm{NMSE}_{mk}(\thetv,\betv) $ are given in closed-forms by \eqref{nablath1}, \eqref{nablath2},
	\begin{figure*}
		\begin{align}
			\nabla_{\thetv}	\mathrm{NMSE}_{mk}(\thetv,\betv) &=[\nabla_{\thetv^{t}}	\mathrm{NMSE}_{mk}(\thetv,\betv)^{\T}, \nabla_{\thetv^{r}}	\mathrm{NMSE}_{mk}(\thetv,\betv)^{\T}]^{\T},\label{nablath1}\\
			\nabla_{\thetv^{i}}\mathrm{NMSE}_{mk}(\thetv,\betv)&=\frac{\tr(\bPsi_{mk})\nabla_{\boldsymbol{\theta}^{i}}\tr(\bR_{mk})-\tr(\bR_{mk})\nabla_{\boldsymbol{\theta}^{i}}\tr(\bPsi_{mk})}{\tr^{2}(\bR_{mk})} , i=t,r. \label{nablath2}
		\end{align}
		\hrulefill
	\end{figure*}
	where
	\begin{align}
		\nabla_{\thetv^{i}}\tr(\bR_{mk})
		&=\tilde{\beta}_{mk}\tr(\mathbf{R}_{\mathrm{AP}})\diag\bigl(\mathbf{A}_{i}\diag(\boldsymbol{{\beta}}^{i})\bigr)\label{derivtheta},\\
		\nabla_{\thetv^{i}}\tr(\bPsi_{mk})&=	\nu_{mk}\diag\bigl(\mathbf{A}_{i}\diag(\boldsymbol{{\beta}}^{i})\bigr)
	\end{align}
	with
	\begin{equation}
		\nu_{mk}\!=\!\hat{\beta}_{mk}\!\tr\bigl(\!\bigl(\mathbf{Q}_{k}\mathbf{R}_{mk}+\mathbf{R}_{mk}\mathbf{Q}_{k}-p \tau \sum_{l\in \mathcal{P}_{k} }\mathbf{Q}_{k}\mathbf{R}_{mk}^{2}\mathbf{Q}_{k}\bigr)\mathbf{R}_{\mathrm{AP}}\bigr).
	\end{equation}

	Similarly, the  gradient $\nabla_{\betv}\mathrm{NMSE}_{mk}(\thetv,\betv) $ is given by \eqref{nablab1}, \eqref{nablab2},
	\begin{figure*}
		\begin{align}
			\nabla_{\betv}	\mathrm{NMSE}_{mk}(\thetv,\betv) &=[\nabla_{\betv^{t}}	\mathrm{NMSE}_{mk}(\thetv,\betv)^{\T}, \nabla_{\betv^{r}}	\mathrm{NMSE}_{mk}(\thetv,\betv)^{\T}]^{\T},\label{nablab1}\\
			\nabla_{\betv^{i}}\mathrm{NMSE}_{mk}(\thetv,\betv)&=\frac{\tr(\bPsi_{mk})\nabla_{\boldsymbol{\beta}^{i}}\tr(\bR_{mk})-\tr(\bR_{mk})\nabla_{\boldsymbol{\beta}^{i}}\tr(\bPsi_{mk})}{\tr^{2}(\bR_{mk})} , i=t,r. \label{nablab2}
		\end{align}
		\hrulefill
	\end{figure*}
	where
	\begin{align}
		\!\!	\nabla_{\betv^{i}}\tr(\bR_{mk})		&\!=\!2\tilde{\beta}_{mk}\tr(\mathbf{R}_{\mathrm{AP}})\Re\bigl\{\diag\bigl(\mathbf{A}_{i}\herm\diag(\btheta^{i})\bigr\},\\
		\!\!	\nabla_{\betv^{i}}\tr(\bPsi_{mk})&\!=\!2	\nu_{mk}\diag\bigl(\mathbf{A}_{i}\Re\bigl\{\diag\bigl(\mathbf{A}_{i}\herm\diag(\btheta^{i})\bigr\}.
	\end{align}
\end{proposition}
\begin{proof}
	Please see Appendix~\ref{lem2}.	
\end{proof}

\subsection{Complexity Analysis of Algorithm \ref{Algoa1}}
Herein, we present the complexity analysis \textcolor{black}{for each iteration} of Algorithm \ref{Algoa1}   using the big-O notation.  First, we focus on the computation of $\bR_{mk}$. We observe that $\bR_{\mathrm{RIS},m} \bPhi_{w_{k}}$ requires $N^2$ complex multiplications because $\bPhi_{w_{k}}$ is diagonal. Thus, to compute  $\tr(\bA_{w_k} \bPhi_{w_{k}}^{\H})$, the complexity is $O(N^2+N)$. The complexity to compute $\bR_{mk}$ is $O(N^2+N+L^2)$ because $O(L^2)$ additional complexity multiplications are required to derive $\tr(\bA_{w_k} \bPhi_{w_{k}}^{\H})\bR_{\mathrm{AP},m}$. Moreover, since $\bPsi_{mk}=\bR_{mk}\bQ_{k}\bR_{mk}$ with $ \bQ_{k} $ being inverse takes $O(L^3)$ to derive it. In summary, we can conclude that the complexity \textcolor{black}{for each iteration} is $	\mathrm{NMSE}_{mk}(\thetv,\betv)$ is $O(K M(N^2+L^3))$.

	{\color{black}
	\subsubsection{Convergence Analysis of Algorithm \ref{Algoa1}}
	The guarantee of the convergence of Algorithm \ref{Algoa1} is provided by following standard arguments for projected gradient methods. First, the gradients $\nabla_{\thetv}f(\thetv,\betv)$ and $\nabla_{\betv}f(\thetv,\betv)$ are Lipschitz continuous\footnote{\color{black}A function $\bh(\bx)   $ is said to be Lipschitz continuous over the set $D$ if there exists $L>0$ such that $||\bh(\bx)-\bh(\by)  ||\leq L||\bx-\by||_2$} over the feasible set as they comprise basic functions as given above. Let $L_{\thetv }$ and $L_{\betv}$ be the Lipschitz constant of $\nabla_{\thetv}f(\thetv,\betv)$ and $\nabla_{\betv}f(\thetv,\betv)$, respectively. Next, we have that \cite[Chapter 2]{Bertsekas1999}
	\begin{align}
		f(\bx,\by) &\geq f(\thetv,\betv)
		+\langle	\nabla_{\thetv}f(\thetv,\betv),\bx-\thetv\rangle-\frac{1}{L_{\thetv }}\|\bx-\thetv\|^{2}_{2}
	\nn	\\
		&\quad\quad+\langle\nabla_{\betv}f(\thetv,\betv),\by-\betv\rangle-\frac{1}{L_{\betv}}\|\by-\betv\|^{2}_{2}\nn\\
		&\geq f(\thetv,\betv)
		+\langle	\nabla_{\thetv}f(\thetv,\betv),\bx-\thetv\rangle-\frac{1}{L_{\max }}\|\bx-\thetv\|^{2}_{2}
		\nn\\ \nn&\quad\quad+\langle\nabla_{\betv}f(\thetv,\betv),\by-\betv\rangle-\frac{1}{L_{\max}}\|\by-\betv\|^{2}_{2}		
	\end{align}
	where $L_{\max}=\max(L_{\thetv },L_{\betv})$. Hence, the line search procedure of Algorithm \ref{Algoa1} (i.e. the loop between Steps \ref{ls:start} -- \ref{ls:end}) terminates in finite iterations since the condition in Step \ref{ls:end} must be satisfied when  $\mu_n <L_{\max}$. More specifically, given $\mu_{n-1}$, the maximum number of steps in the line search procedure is $\left\lceil \frac{\log(L_{\max}\mu_{n-1})}{\log\kappa}\right\rceil $. Moreover, because of the line search we automatically have  an increasing sequence of objectives, i.e., $f(\thetv^{n+1},\betv^{n+1})\geq f(\thetv^{n},\betv^{n})$. Since the feasible sets $\Theta$  and $\mathcal{B}$ are compact, $f(\thetv^{n},\betv^{n})$ must converge. However, we highlight that Algorithm \ref{Algoa1} is only guaranteed to converge to a stationary point of \eqref{Maximization}, which is not necessarily an optimal solution due to the nonconvexity of \eqref{Maximization}. We also note that $L_{\thetv }$ and $L_{\betv}$ are not required to run Algorithm \ref{Algoa1}.}

\section{Downlink Data Transmission}\label{PerformanceAnalysis}

This section presents the downlink data transmission phase, and the study of the received signal  when both the numbers of APs and RIS elements grow to infinity. Also, we provide a closed-form expression of the achievable downlink SE with MR precoding for an arbitrary PB.

Based on TDD, we can exploit channel reciprocity, where the uplink and downlink channels are the same, and write the received signal by UE $ k $ in  $ t  $ or $ r  $ region. Specifically, we consider the cooperation among the $ M $ APs that jointly transmit the same data symbol to UE $ k $. In particular, the received signal by  UE $ k $ is described as
\begin{align}
	r_{k}=\sum_{m=1}^{M}\bh_{mk}^\H\bs_{m}+z_{k},\label{DLreceivedSignal}
\end{align}
where   $\bs_{m}=\sqrt{\rho_{\mathrm{d}}}  \sum_{i=1}^{K}\sqrt{\eta_{mi}}\bff_{i}l_{i}$ denotes  the transmit signal vector  by the $ m $th AP  with $ \rho_{\mathrm{d}} $ being  the normalized SNR in the downlink allocated to UE $k $,  and $z_{k} \sim \cC\cN(0,1)$ being the complex Gaussian noise at UE $k$. Also,  $\bff_{k} \in \bbC^{L \times 1}$ is the linear precoding vector, and   $ l_{k} $ is     the corresponding data symbol with $ \EE\{|l_{i}|^{2}\}=1 $. After accounting for MR precoding, where $ \bff_{i}= \hat{\bh}_{mi} $, \eqref{DLreceivedSignal} can be written as
\begin{align}
	r_{k}=\sqrt{\rho_{\mathrm{d}}}\sum_{m=1}^{M}  \sum_{i=1}^{K}\sqrt{\eta_{mi}}{\bh}_{mk}^\H \hat{\bh}_{mi} l_{i} +z_{k}.\label{DLreceivedSignal1}
\end{align}
Note that $ \eta_{mk} $ is a power control coefficient at AP $ m $ that satisfies the
power  constraint $\EE\{\|\bs_{m}\|^{2}\}\le \rho_{\mathrm{d}}  $, which gives
\begin{align}
	\sum_{k=1}^{K}\eta_{mk}\tr(\bPsi_{mk})\le 1.\label{constraint}
\end{align}

\subsection{Asymptotic Analysis ($ M, N \to \infty $)}
To proceed with the analysis in this case, certain assumptions, concerning the covariance matrices, should be fulfilled \cite[Assump. A1-A3]{Hoydis2013}. These assumptions basically mean that the  sum of the eigenvalues and the largest singular value of the covariance matrices are finite and positive. As can be seen in \eqref{DLreceivedSignal1}, the received signal depends on the channel estimates of all UEs. To proceed further, this equation is rewritten  in terms of the  pilot reuse set as
\begin{align}
	r_{k}&=\sqrt{\rho_{\mathrm{d}}}\sum_{i \in \mathcal{P}_{k}}\sum_{m=1}^{M}  \sqrt{\eta_{mi}}{\bh}^\H_{mk}\hat{\bh}_{mi}l_{i}\nn\\
	&+\sqrt{\rho_{\mathrm{d}}}\sum_{i\not\in\mathcal{P}_{k} }\sum_{m=1}^{M}  \sqrt{\eta_{mi}}{\bh}^\H_{mk}\hat{\bh}_{mi}l_{i} +z_{k}.\label{DLreceivedSignal2}
\end{align}

Elaborating on the first sum of \eqref{DLreceivedSignal2}, we have
\begin{align}
	&	\sum_{m=1}^{M}  \sqrt{\eta_{mi}}{\bh}^\H_{mk}\hat{\bh}_{mi}\nn\\
	&	=\sum_{m=1}^{M}  \sqrt{\eta_{mi}}{\bh}^\H_{mk} \sqrt{p \tau}\bR_{mi}\bQ_{m} \bigl(\sum_{l\in \mathcal{P}_{k} }\sqrt{p \tau}\bh_{ml}+\bz_{\mathrm{p},mi}\bigr)
	\label{received2}\\
	&	=\sum_{m=1}^{M}  \sqrt{\eta_{mi}p \tau}\bh^\H_{mk} \bR_{mi}\bQ_{m} \bh_{mk}+\sum_{l\in \mathcal{P}_{k} \not \in \{k\}}\sum_{m=1}^{M}  \sqrt{\eta_{mi}p \tau}\nn\\
	&\times{\bh}^\H_{mk} \bR_{mi}\bQ_{m}\bh_{ml}
	+\sum_{m=1}^{M}  \sqrt{\eta_{mi}}{\bh}^\H_{mk} \bR_{mi}\bQ_{m}\bz_{\mathrm{p},mi},\label{received3}
\end{align}
where \eqref{received2} is derived by substituting the channel estimate provided by \eqref{estim1}, while  \eqref{received3} is obtained after extracting the channel of UE $ k $ from the summation. When $ M, N\to \infty $, the asymptotic result is written after dividing each term by $ MN $ as
\begin{align}
	&\frac{1}{MN}	\sum_{m=1}^{M}\sqrt{\eta_{mi}p \tau}\bh^\H_{mk} \bR_{mi}\bQ_{m} \bh_{mk} \nn\\
	&=\frac{1}{MN}	\sum_{m=1}^{M}\sqrt{\eta_{mi}p \tau}\bq^\H_{k}\bPhi_{w_{k}}^{\H}\bW_{m} \bR_{mi}\bQ_{m} \bW_{m}^{\H}\bPhi_{w_{k}} \bq_{k}\nn\\
	&=\frac{1}{N}\sqrt{ \tau}\bq^\H_{k}\bPhi_{w_{k}}^{\H}\Big(\frac{1}{M}	\sum_{m=1}^{M}\sqrt{\eta_{mi}}\bW_{m} \bR_{mi}\bQ_{m} \bW_{m}^{\H}\Big)\bPhi_{w_{k}} \bq_{k}\nn\\
	&\xrightarrow[M\to \infty]{P}\sqrt{p \tau}\nn\\
	&\times\frac{1}{N}\bq^\H_{k}\bPhi_{w_{k}}^{\H}\frac{1}{M}\!\! \sum_{m=1}^{M}\!\!\sqrt{\eta_{mi}}\EE\{\bW_{m} \bR_{mi}\bQ_{m}\bW_{m}^{\H}\}\bPhi_{w_{k}} \bq_{k}\nn\\
	&=\sqrt{p \tau}\nn\\
	&\times\frac{1}{N}\bq^\H_{k}\bPhi_{w_{k}}^{\H}\frac{1}{M}\!\! \sum_{m=1}^{M}\!\!\sqrt{\eta_{mi}}\tilde{ \beta}_{m}\bR_{\mathrm{RIS},k}\tr( \bR_{mi}\bQ_{m}\bR_{\mathrm{AP},m})\bPhi_{w_{k}} \bq_{k}\nn\\
	&=\sqrt{p \tau}\frac{1}{N}\bq^\H_{k}\bA\bq_{k},\label{received4}
\end{align}
where $ \bA= \tr( \bR_{mi}\bQ_{m}\bR_{\mathrm{AP},m})\frac{1}{M}\!\!\sum_{i \in \mathcal{P}_{k}} \sum_{m=1}^{M}\!\!\sqrt{\eta_{mi}}\tilde{ \beta}_{m}\bPhi_{w_{k}}^{\H}$   $\bR_{\mathrm{RIS},k}\bPhi_{w_{k}}$ we have used Tchebyshev's theorem \cite{Cramer2004}, while the second and third terms tend to zero due to favorable propagation conditions, and because  the overall channel and the noise  are mutually independent. This result shows that, for a fixed $ N $, the channels become asymptotically  orthogonal. Thus, the small-scale fading and the additive noise cancel out. The received signal becomes 
\begin{align}
	\!\!	\frac{1}{MN}	r_{k}\!\!&\xrightarrow[M\to \infty]{P}\!\sqrt{p \tau}\frac{1}{N}\bq^\H_{k}\bA\bq_{k}\nn\\
	&\xrightarrow[M\to \infty]{P}\sqrt{p \tau}\tilde{ \beta}_{k}\frac{1}{N}\tr(\bA \bR_{\mathrm{RIS},k}),\label{received7}
\end{align}
where the contribution of the  STAR-RIS appears indirectly in \eqref{received7}. However, the pilot contamination from UEs using the same pilot sequence remains, which means that the system cannot be benefited if we add more APs.

%
%

\subsection{Finite Analysis}
By taking advantage of the hardening channel capacity bounding technique,  the downlink ergodic spectral efficiency in $ \mathrm{bps/Hz} $ can be written as
\begin{align}
	\mathrm{SE}_{k} =(1-\tau/\tau_{\mathrm{c}}) \log_{2 }(1+ \gamma_{k}),
\end{align}
where   the pre-log fraction expresses  the percentage of samples per coherence block  for downlink data transmission, and the effective uplink signal-to-interference-plus-noise ratio (SINR) $  \gamma_{k}  $ is given by
\begin{align}
	\gamma_{k}=	\frac{S_{k}}{	I_{k}},
	\label{SINR}
\end{align}
with \begin{align}
	S_{k}&=|\mathrm{DS}_{k}|^{2}\label{num}\\
	I_{k}&=\EE\{|\mathrm{BU}_{k}|^{2}\}+\!\!\sum_{i=1, i\ne k}^{K}\!\!\EE\{|\mathrm{UI}_{ik}|^{2}\}+1\label{denom}.
\end{align}

In \eqref{num}, $ \mathrm{DS}_{k} $ is the desired signal, while, in \eqref{denom}, $ \mathrm{BU}_{k} $ is the beamforming gain, and $ \mathrm{UI}_{ik} $ is the multi-UE interference. Specifically, we have
\begin{align}
	\mathrm{DS}_{k}&\!=\!\sqrt{\rho_{\mathrm{d}}}\EE\Big\{\sum_{m=1}^{M}  \sqrt{\eta_{mk}}{\bh}^\H_{mk}\hat{\bh}_{mk}\Big\},\\
	\mathrm{BU}_{k}&\!=\!\sqrt{\rho_{\mathrm{d}}}\Big(\!\sum_{m=1}^{M}  \sqrt{\eta_{mk}}{\bh}^\H_{mk}\hat{\bh}_{mk}\!-\!\EE\Big\{\!\!\sum_{m=1}^{M} \!\! \sqrt{\eta_{mk}}{\bh}^\H_{mk}\hat{\bh}_{mk}\!\Big\}\!\!\Big),\\
	\mathrm{UI}_{ik}&\!=\!\sum_{m=1}^{M}  \sqrt{\eta_{mi}}{\bh}^\H_{mk}\hat{\bh}_{mi}.
\end{align}

Proposition \ref{Proposition:DLSINR} provides a closed-form representation of \eqref{SINR}.

\begin{proposition}\label{Proposition:DLSINR}
	For a given PB $ \Phi_{w_k} $ and  MR precoding being used, the downlink achievable SINR of UE $k$  in a STAR-RIS-assisted CF mMIMO system  is given by  \eqref{SINR}, where
	\begin{align}
		S_{k}	&=\rho_{\mathrm{d}}(\sum_{m=1}^{M}  \sqrt{\eta_{mk}}\tr(\bPsi_{mk}))^{2},\label{ds3}
	\end{align}
	\begin{align}
		&I_{k} =\rho_{\mathrm{d}}\sum_{m=1}^{M} \sum_{n=1, n\ne m}^{M} \sqrt{\eta_{mk}} \sqrt{\eta_{nk}} (p \tau)^{2} \tr(\bR_{\mathrm{AP},m}\bR_{mk}\bQ_{m})\nn\\
		&\times\tr(\bR_{\mathrm{AP},m}\bQ_{n}\bR_{nk})
		\big(\tr(\bR_{mk})\tr(\bR_{nk})+\tr(\bR_{mk}\bR_{nk})\big)\nn\\
		&+(p \tau)^{2}\!\!\!\!\sum_{l\in \mathcal{P}_{k}\backslash{k} }\!\!\!\! \tr(\bR_{\mathrm{AP},m}\bR_{mk}\bQ_{m})\!\tr(\bR_{\mathrm{AP},m}\bQ_{n}\bR_{nk})	\! \tr(\tilde{\bR}_{mk}\tilde{\bR}_{nl})\nn\\
		&-\rho_{\mathrm{d}} \sqrt{\eta_{mk}} \sqrt{\eta_{nk}}\tr(\bPsi_{mk})\tr(\bPsi_{nk})\nn\\
		&+ 	 	 \rho_{\mathrm{d}}\!\!\sum_{m=1}^{M}\!\eta_{mk}(p \tau)^{2} \tr(\bR_{\mathrm{AP},m}\bR_{mk}\bQ_{m}\bR_{\mathrm{AP},m}\bQ_{m}\bR_{mk})\nn\\
		&	\times\Big(\!\!\tr\big(\tilde{\bR}_{mk}^{2}\big)+\big(\!\tr\big(\tilde{\bR}_{mk}\big)\!\big)^{2}\!\Big)\nn\\
		&+	(p \tau)^{2}\sum_{l\in \mathcal{P}_{k}\backslash{k}}\Big(\tr(\tilde{\bR}_{ml})
		\nn\\
		&		\times\tr(\bR_{\mathrm{AP},m}\bR_{mk}\bQ_{m}\bR_{\mathrm{AP},m}\bQ_{m}\bR_{mk})\tr\big(\tilde{\bR}_{mk}\big)\Big)^{2}\nn\\
		&+	p \tau \tr(\bR_{mk}^{2}\bQ_{m}^{2})\!-\!\rho_{\mathrm{d}}\!\!\sum_{m=1}^{M}\!\eta_{mk}\tr^{2}(\bPsi_{mk})\nn\\
		&+	p \tau\!\!\!\sum_{i=1, i\ne k}^{K}\!\!\!\!\bigg(\sum_{m=1}^{M}  {\eta_{mi}}\tr(\bR_{mi}\bQ_{m}^{2}\bR_{mi}\bR_{mk})\nn\\
		&+(p \tau)^{2}\left\{
		\begin{array}{ll}
			\tr(\bR_{\mathrm{AP},m}\bR_{mi}\bQ_{m}\bR_{\mathrm{AP},m}\bQ_{m}\bR_{mi})\tr(\tilde{\bR}_{ml})\tr(\tilde{\bR}_{mk})\\+	\sum_{m=1}^{M}\sum_{n\ne m}^{M}\sqrt{\eta_{mi}\eta_{ni}}\sum_{l\in \mathcal{P}_{i}}\tr(\bR_{\mathrm{AP},m}\bR_{mi}\bQ_{m})\nn\\
			\\\times	\tr(\bR_{\mathrm{AP},m}\bQ_{n}\bR_{ni})	 \tr(\tilde{\bR}_{mk}\tilde{\bR}_{nl})\!\!\bigg) &\!\!\!\!\!\!\!\!\!\!\!\!\!\!\!\!\!\!\!\!\!\!\!\!\!\!\!\!\!\!\!\!\!\!\!\!\! i\not\in \mathcal{P}_{k} \\
			\tr(\bR_{\mathrm{AP},m}\bR_{mi}\bQ_{m}\bR_{\mathrm{AP},m}\bQ_{m}\bR_{mi})\\
			\nn\\
			\times\Big(\!\!\tr\big(\tilde{\bR}_{mk}^{2}\big)+\big(\!\tr\big(\tilde{\bR}_{mk}\big)\!\big)^{2}\!\Big)\nn\\
			+
			\sum_{l\in \mathcal{P}_{k}\backslash{k}}\tr(\bR_{\mathrm{AP},m}\bR_{mi}\bQ_{m}\bR_{\mathrm{AP},m}\bQ_{m}\bR_{mi})\\
			\nn\\
			\times\tr(\tilde{\bR}_{ml})\tr(\tilde{\bR}_{mk})\nn\\
			+\tr(\bR_{\mathrm{AP},m}\bR_{mi}\bQ_{m})\tr(\bR_{\mathrm{AP},m}\bQ_{n}\bR_{ni})\\
			\nn\\
			\times\big(\tr(\bR_{mk})\tr(\bR_{nk})+\tr(\bR_{mk}\bR_{nk})\big)\nn\\
			+\sum_{l\in \mathcal{P}_{k}\backslash{k}}\tr(\bR_{\mathrm{AP},m}\bR_{mi}\bQ_{m})\\
			\nn\\
			\times\tr(\bR_{\mathrm{AP},m}\bQ_{n}\bR_{ni}) \tr(\tilde{\bR}_{mk}\tilde{\bR}_{nl})\!\!\bigg). &\!\!\!\!\!\!\!\!\!\!\!\!\!\!\!\!\!\!\!\!\!\!\!\!\!\!\!\!\!\!\!\!\!\!\! i\not\in \mathcal{P}_{k}\backslash{k}\\
		\end{array} 
		\right. .
	\end{align}
	
\end{proposition} 
\begin{proof}
	See Appendix~\ref{Proposition1}.	
\end{proof}

The achievable sum SE is obtained as 
\begin{align}
	\mathrm{SE}	=\frac{\tau_{\mathrm{c}}-\tau}{\tau_{\mathrm{c}}}\sum_{k=1}^{K}\log_{2}\left ( 1+\gamma_{k}\right)\!,\label{LowerBound}
\end{align}
where  $ \gamma_{k}$ is given by Proposition \ref{Proposition:DLSINR}.
\begin{remark}
	According to Proposition \ref{Proposition:DLSINR}, the downlink achievable SINR is given in closed-form and depends only on statistical CSI in terms of path losses and covariance matrices. We have chosen to optimize the amplitudes and the phase shifts of the STAR-RIS by minimizing the total NMSE as mentioned in Section \ref{ChannelEstimation}. The optimization of the STAR-RIS by maximizing the achievable sum SE is omitted due to limited space but will be the topic of future work.
\end{remark}

\section{Numerical Results}\label{Numerical}
In this section, we present the numerical results of the sum SE in  STAR-RIS-aided CF mMIMO systems, which include analytical results and Monte-Carlo (MC) simulations with $ 10^{3} $ independent channel realizations. 

The  setup assumes  a geographic area of size $  1.5 \times 1.5 $ $ \mathrm{km}^{2} $, where the locations of all nodes are given in terms of
$ (x, y) $ coordinates. In particular, we consider a STAR-RIS with a UPA of $ N=64 $ elements deployed on the wall of a building. The STAR-RIS aids the communication between  $ M =64$ randomly located APs antennas  with $ L=4 $ each that serve $ K_{t} = 4 $ indoor UEs and $ K_{r}=3 $ outdoor UEs. Specifically, the $xy-$coordinates of the APs are  uniformly distributed around $(x_0,~ y_0) = (0,~0)$ , while the STAR-RIS is located  at $(x_R,~ y_R)=(50,~ 10)$, all in meter units.  Also, UEs in $r$ region are located on a straight line between $(x_R-\frac{1}{2}d_1,~y_R-\frac{1}{2}d_1)$ and $(x_R+\frac{1}{2}d_1,~y_R-\frac{1}{2}d_1)$ with equal distances between each two adjacent users, and $d_1 = 20$~m in our simulations. In a similar way, UEs in the $t$ region are located between $(x_R-\frac{1}{2}d_2,~y_R+\frac{1}{2}d_2)$ and $(x_R+\frac{1}{2}d_2,~y_R+\frac{1}{2}d_2)$ with $d_2 = 1$~m. We assume that the size of each RIS element is $ d_{\mathrm{H}}\!=\!d_{\mathrm{V}}\!=\!\lambda/4 $.   Distance-based path-loss is considered in our work, such that the channel gain of a given link $j$ is $\tilde \beta_j = A d_j^{-\alpha_{j}}$, where \textcolor{black}{ $A=d_{\mathrm{H}}\times d_{\mathrm{V}}$} is the area of each reflecting element at the RIS, \textcolor{black}{$ d_j $ is the distance of the corresponding link}, and $\alpha_{j}=2.5$ is the path-loss exponent.  Note that $ d_{ \mathrm{H}}=d_{ \mathrm{V}}=\lambda/4 $ unless otherwise stated. The correlation matrices $ \bR_{\mathrm{BS}}$ and $\bR_{\mathrm{RIS},m} $, $\bR_{\mathrm{RIS},k} $ are  computed according to \cite{Hoydis2013} and \cite{Bjoernson2020}, respectively.   The carrier frequency and the system bandwidth are  $ 1.9 $ $ \mathrm{GHz} $  and  $ 20 $ $ \mathrm{MHz} $, respectively. Each 
coherence interval consists of $ \tau_{\mathrm{c}} = 200 $ symbols, which 
correspond to a coherence bandwidth equal to $ B_{\mathrm{c}}  = 200$ $ \mathrm{KHz} $ and a coherence time equal to $ T_{\mathrm{c}} = 1 $ $ ms $. We consider $ \tau = 5 $ orthonormal pilot sequences that are shared by all UEs. \textcolor{black}{ We would like to mention that we consider equal power allocation as commonly  assumed  in  mMIMO systems. Hence, we have assumed $ \eta_{mk}=(\sum_{k'=1}^{K}\tr(\bPsi_{mk'}))^{-1}, \forall m, k $ from \eqref{constraint} Optimal  power control such as max-min power control will take place in a future work.}

For the evaluation of the advantages  of RIS-assisted CF  mMIMO systems, we consider  the following scenarios  for comparison:
\begin{itemize}
	\item A conventional RIS :  This is a baseline scheme, which consists of transmitting-only or reflecting-only elements, each with $ N_{t} $ and $ N_{r} $ elements, such that $ N_{t}+N_{r} =N$. We denote it ``cRIS''.
	\item We apply an ON/OFF scheme for  channel estimation based on \cite{Mishra2019}, where the cascaded links are estimated with one  element turned on at transmission/reflection mode sequentially. We denote it ``ON/OFF scheme''.
	\item A random PB, where the phase shifts and the amplitudes are chosen based on the Uniform distribution. ``random PB''.
	\item A conventional CF mMIMO without any surface. We denote it ``cCF mMIMO''.
	\item \textcolor{black}{An active STAR-RIS as provided in \cite{Xu2023}}.
\end{itemize}

\textcolor{black}{Fig. \eqref{fig21} illustrates the total relative estimation error, i.e., the normalized mean square error (NMSE) with 
respect to the uplink SNR for different PB matrices.
We show that the results decrease without bound. Moreover,
it is shown that the error  goes with the ON/OFF scheme in \cite{Mishra2019} and the scenario of equal phase shifts. In addition, we have  provided a comparison between the ES and the MS protocols, where the former presents lower error while the gap increases with increasing SNR. The reason for this observation is that at low SNR, under our considered setup, it is more beneficial to focus on the UEs in the reflection region as they are closer to the APs. This is confirmed by the fact that, after
running the proposed algorithm, $ \beta_{n}^{r}\approx 1,  \forall n \in \mathcal{N} $.  For high SNR, instead of pouring all the power to the users in the
reflection region, some power can be directed to the users in the transmission region to improve the total NMSE. Notably, MC simulations verify the analytical results.}

\begin{figure}[!h]
	\begin{center}
		\includegraphics[width=0.8\linewidth]{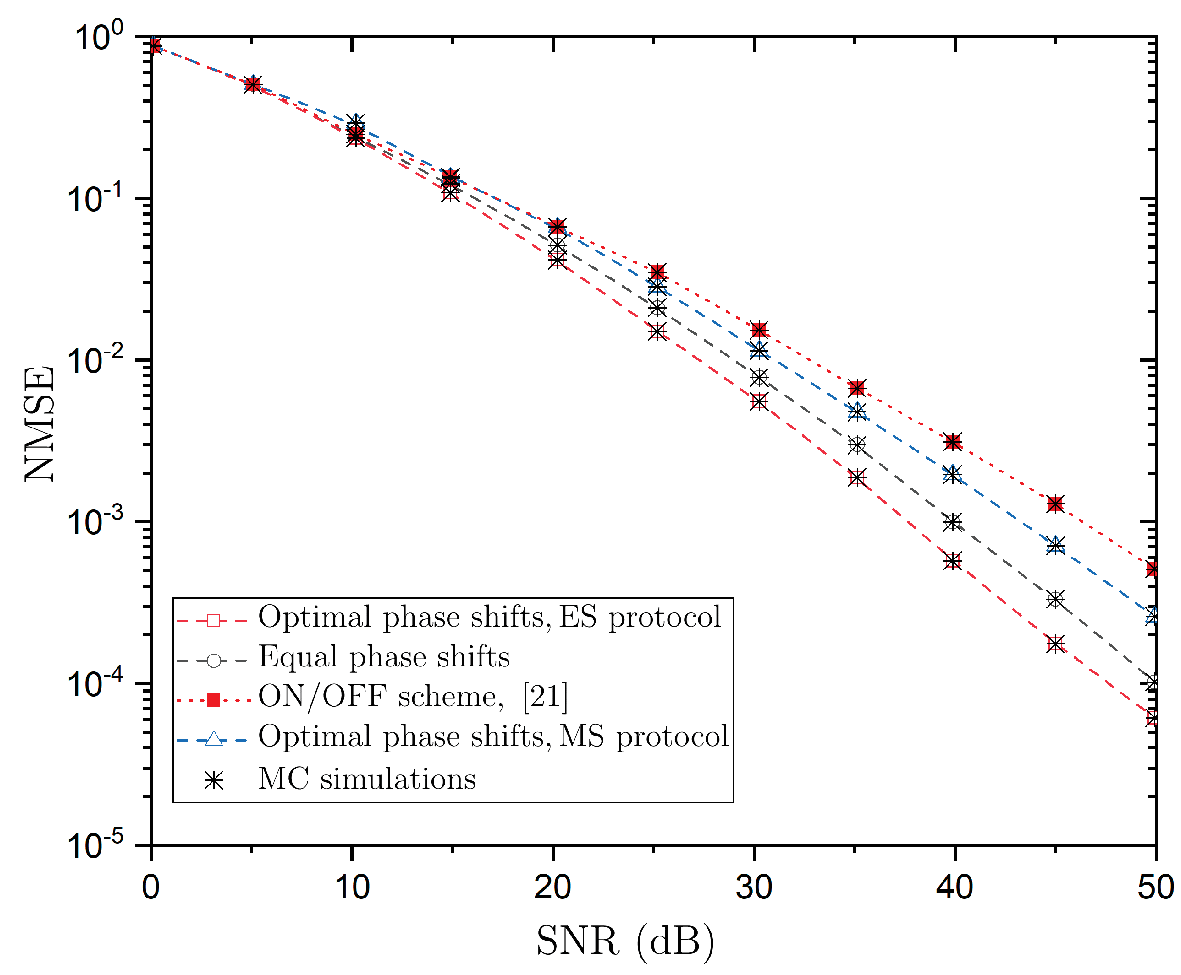}
		\caption{ \textcolor{black}{Total NMSE versus the SNR of a STAR-RIS assisted CF mMIMO system
			with imperfect CSI ($ M=64 $, $ L=4 $, $ K=4 $) in the case of uncorrelated fading 			at the RIS (Analytical results and MC simulations). }}
		\label{fig21}
	\end{center}
\end{figure}

In Fig. \ref{fig2}, we illustrate the  achievable data rate as a function of the number of STAR-RIS  elements $ N $. It can be seen that an increase in the number of surface elements brings a significant performance  to cRIS, which motivates the  deployment of a surface in CF mMIMO systems. However, the rate enhancement is smaller in the case of random PB. Moreover, we show the impact of spatial correlation at the STAR-RIS. It is shown that the performance increases as the correlation decreases. Also, the ES protocol exhibits a better performance than the MS protocol since the latter is a special case of the ES protocol. Specifically, the lines coincide for low $ N $. \textcolor{black}{Given that the ES protocol is a full-dimension transmission and reflection scheme, the gap between the ES and the MS protocols increases with increasing $ N $.} Regarding cRIS, the comparison reveals that its performance is lower since fewer degrees of freedom can be taken into advantage.  \textcolor{black}{ Notably, we have added a comparison with active STAR-RIS~\cite{Xu2023}, and have used it as a benchmark. As expected, active STAR-RIS performs better, but at the cost of higher energy consumption.}
\begin{figure}[!h]
	\begin{center}
		\includegraphics[width=0.8\linewidth]{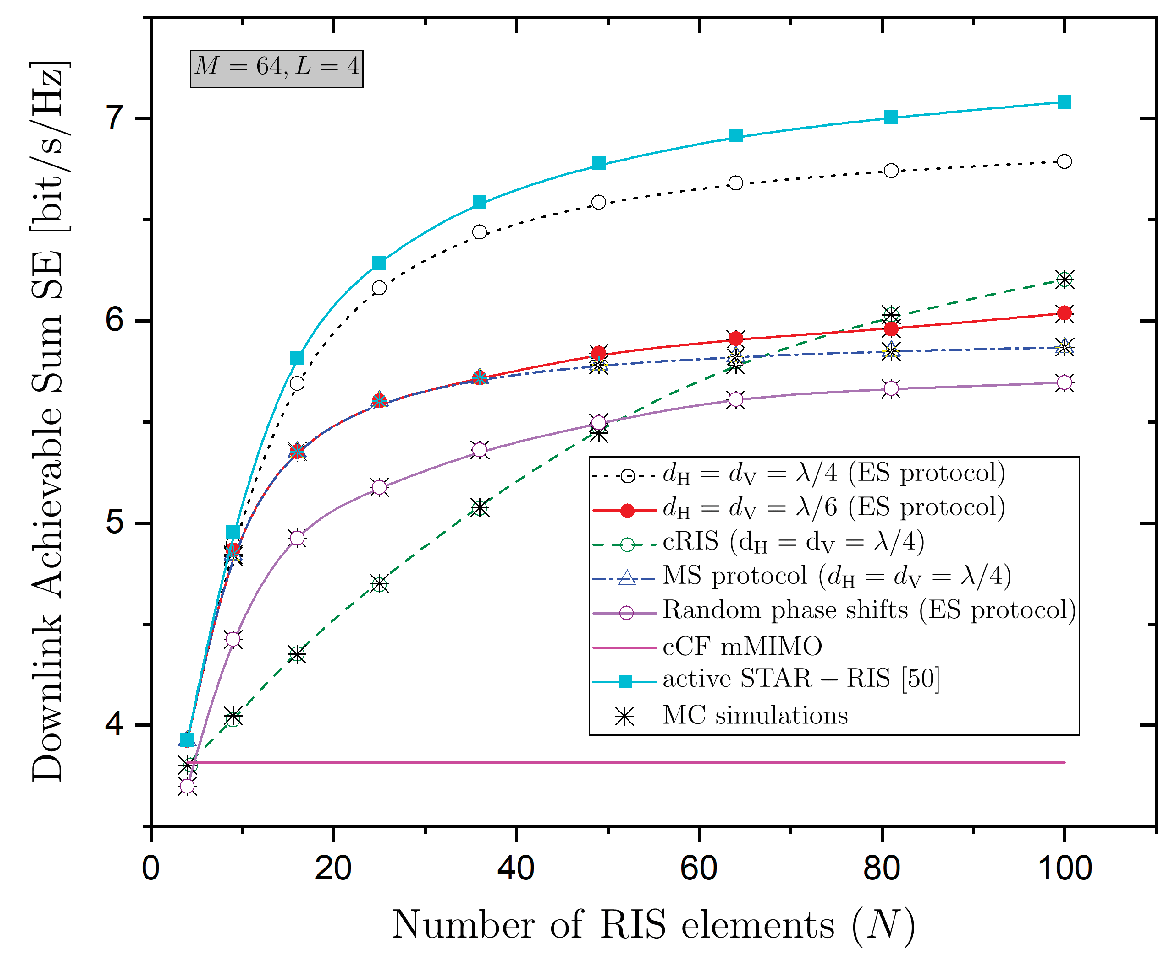}
		\caption{{Downlink achievable sum SE versus the number of RIS elements antennas $N$ of a STAR-RIS assisted MIMO system with imperfect CSI ($ M=64 $, $ L=4 $, $ K=4 $) for varying conditions (Analytical results and MC simulations). }}
		\label{fig2}
	\end{center}
\end{figure}

In Fig. \ref{fig4}, we depict the achievable rate versus the transmit SNR. We show that  even the choice  of a simple MR precoder can enhance the performance in STAR-RIS CF mMIMO systems in the low SNR regime, but as SNR increases, cCF mMIMO will perform better in terms of the rate in the case of random PB. As SNR increases further CF mMIMO can perform better even from the optimal STAR-RIS CF mMIMO system. The reason is the additional multi-user interference  coming from the surface. This observation indicates that the RIS is more beneficial in the low SNR regime. Also, we observe that the ES protocol performs better than the MS protocol. In particular, at low SNR,  the performances of the ES and MS protocols is nearly the same. In this region,  it is more  advantageous to focus on UEs in the reflection region as they are closer to the APs. This fact is confirmed by noticing that after running the proposed algorithm, $\beta_{n}^{r}\approx 1,\forall n\in \mathcal{N}$. As  SNR increases, the increase in the sum SE is small, if we keep focusing on the reflection region. At high SNR, the sum SE is improved since some power can be directed to the UEs in the transmission region.
\begin{figure}[!h]
	\begin{center}
		\includegraphics[width=0.8\linewidth]{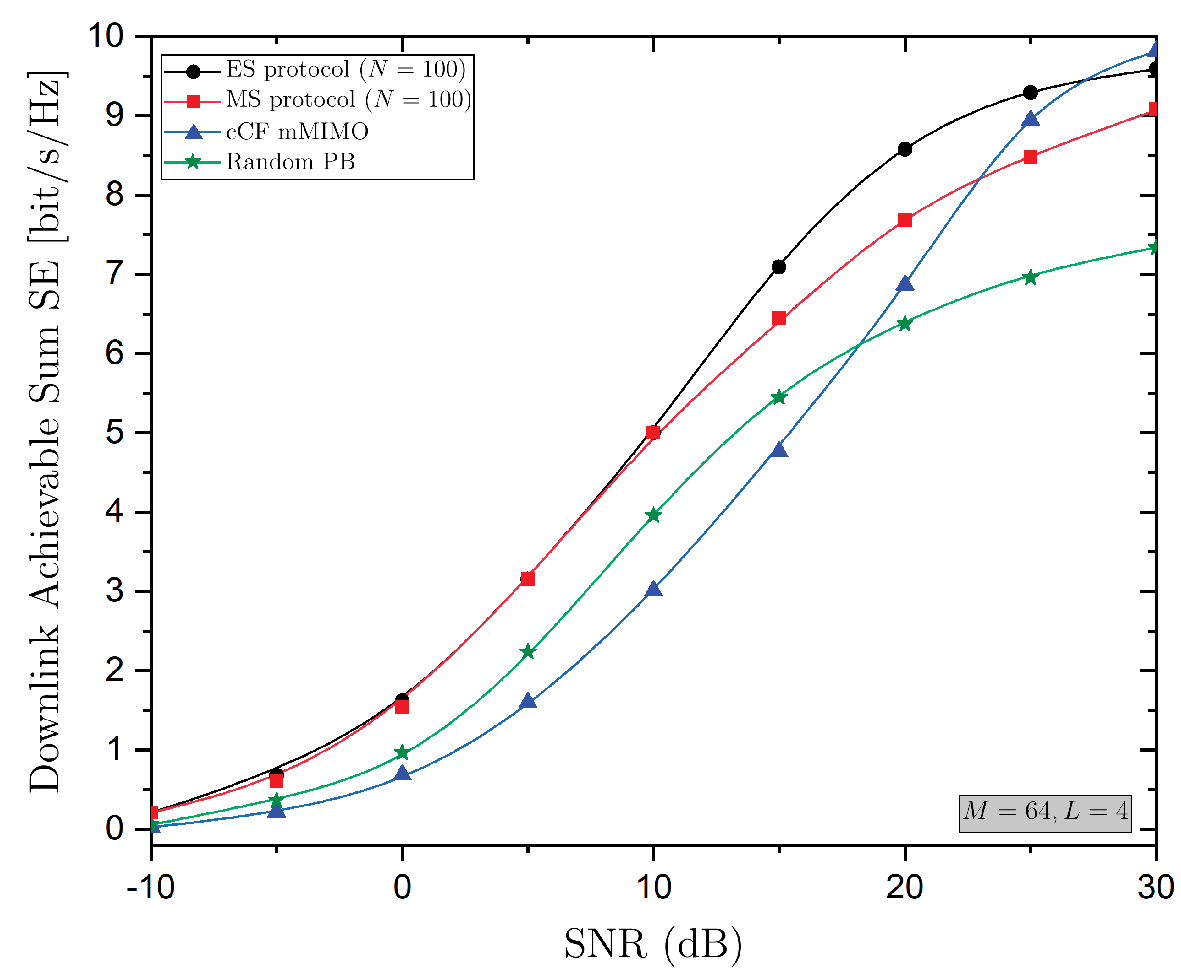}
		\caption{{Downlink achievable sum SE versus the SNR of a STAR-RIS assisted CF mMIMO system with imperfect CSI ($M=64$, $ N=64 $,  $ L=4 $, $ K=4 $) for varying conditions (Analytical results). }}
		\label{fig4}
	\end{center} 
\end{figure}

Fig. \ref{fig3} depicts the achievable rate versus the number of APs $ M $. As shown, the rate increases with $ M $ but saturates at a large $ M $ because of the multi-user interference. The STAR-RIS outperforms the cRIS because it exploits more degrees of freedom due to the simultaneous transmission and reception. Concerning the RIS correlation, a low correlation increases the performance due to increased diversity gains among the surface elements. Despite its higher complexity, the ES protocol achieves better performance than the MS protocol. Under the setting of no surface correlation, the performance is not good because of the lower capability for optimization in terms of only the amplitudes. Also, the scenario corresponding to the ON/OFF scheme presents a higher achievable rate since our case, based on statistical CSI, includes loss of information. Moreover, the scenario with random PB performs  better than the case of no RIS correlation.
\begin{figure}%
	\centering
	\includegraphics[width=0.8\linewidth]{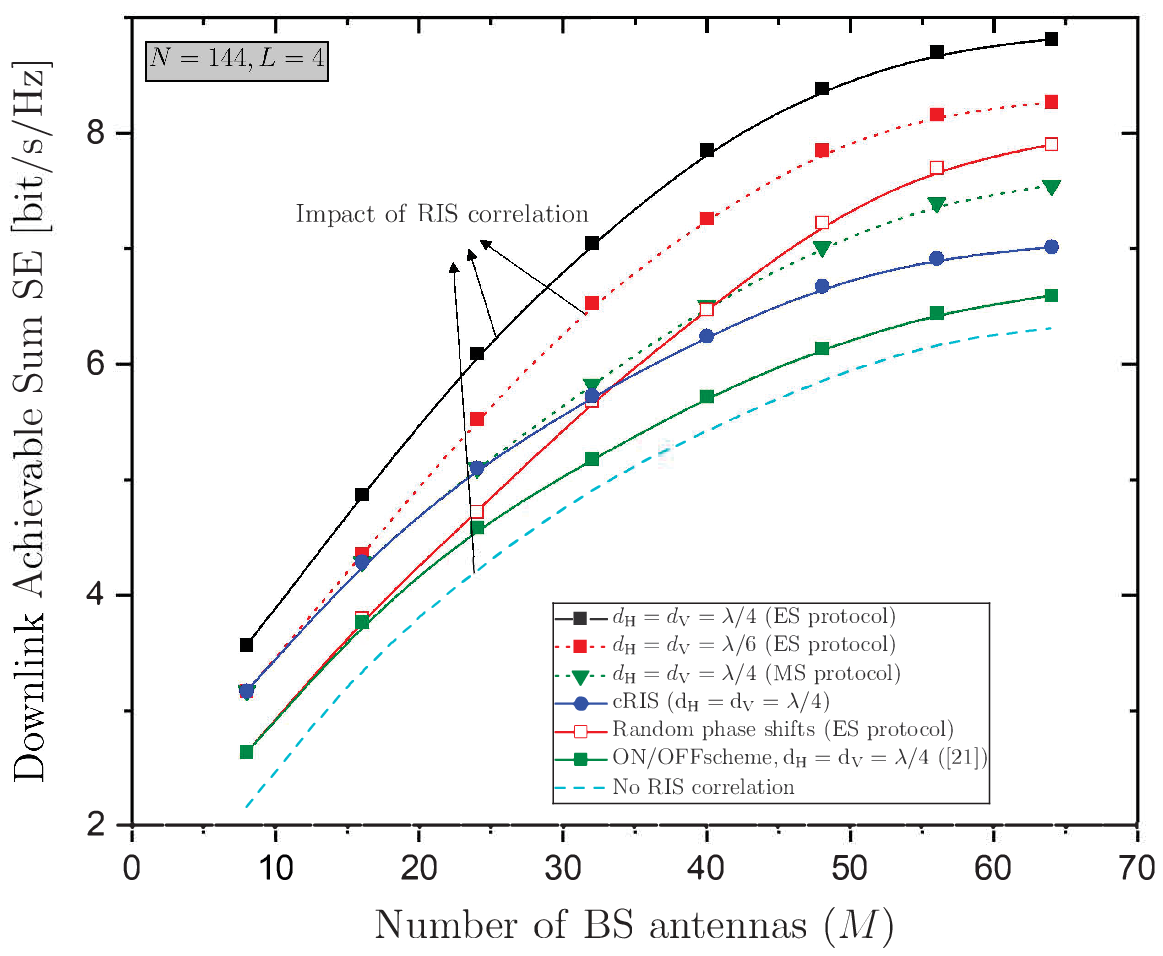}
	\caption{ Downlink achievable sum SE versus the number of BS antennas $M$ of an STAR-RIS assisted MIMO system with imperfect CSI ($ N=64 $, $ L=4 $, $ K=4 $)   under varying conditions (Analytical results). }
	\label{fig3}
\end{figure}
\section{Conclusion}\label{Conclusion}
In this paper, we studied the NMSE and the achievable rate performance in the downlink of STAR-RIS-assisted CF mMIMO systems in terms of statistical CSI. We accounted for correlated Rayleigh fading and introduced a STAR-RIS to provide additional channels to the UEs located on both sides of the surface of a standard CF mMIMO system. We have introduced an efficient channel estimation scheme  to overcome the high overhead obtained by estimating  the individual channels of the RIS elements.  We have employed MR beamforming to design the active beamforming in terms of instantaneous CSI, while the analytical expressions are in terms of statistical CSI. To this end, we have  derived the closed-form expressions of the received signal when the number of APs and surface elements increases, and we have obtained the downlink achievable rate while providing analytical insights. Finally, we have illustrated the numerical
results to demonstrate the  benefits of integrating a STAR-RIS into conventional CF mMIMO systems. Our results have shed light on the impact of spatial correlation and other fundamental parameters on a CF network.

\begin{appendices}
	\section{Proof of Lemma~\ref{lem0proof}}\label{lem0}
	The	computation of the 	expectation in \eqref{exp2} is straightforward based on the property $\bx^{\H}\by = \tr(\by \bx^{\H})$ for any vectors $\bx$, $\by$. In the case of the expectation in \eqref{exp3}, it is written as
	\begin{align}
		&\EE\{ \bh_{mk}^{\H}\bA\bh_{mk}\bh_{nk}^{\H}\bB\bh_{nk}\}
		\nn\\
		&=\EE\{ \bq_{k}^{\H}\bPhi_{w_{k}}^{\H}\bW_{m}\bA\bW_{m}^{\H}\bPhi_{w_{k}}\bq_{k} \bq_{k}^{\H}\bPhi_{w_{k}}^{\H}\bW_{m}\bB\bW_{m}^{\H}\bPhi_{w_{k}}\bq_{k}\},\label{lem1_1}
	\end{align}
	where, in \eqref{lem1_1}, we have substituted the cascaded channel. Note that we have
	\begin{align}
		\bW_{m}\bA\bW_{m}^{\H}&=\EE\{ \tilde{\beta}_{m}\bR_{\mathrm{RIS},m}^{1/2}\bD_{m}\bR_{\mathrm{AP},m}^{1/2}\bA\bR_{\mathrm{AP},m}^{1/2}\bD_{m}^{\H}\bR_{\mathrm{RIS},m}^{1/2}\}\nn\\
		&= \tilde{\beta}_{m}\bR_{\mathrm{RIS},m}^{1/2}\tr(\bR_{\mathrm{AP},m}\bA)\bR_{\mathrm{RIS},m}^{1/2}\nn\\
		&= \tilde{\beta}_{m}\tr(\bR_{\mathrm{AP},m}\bA)\bR_{\mathrm{RIS},m},\label{lem1_11}
	\end{align}
	where, in \eqref{lem1_11}, we have taken into account 
	the property    $ \EE\{\bV \bC\bV^{\H}\} =\tr (\bC) \Id_{M}$, which holds for any holds  $ \bV\in \mathbb{C}^{{M\times N}} $ with
	i.i.d entries of zero mean and unit variance, and $ \bC $ being
	a deterministic square matrix with per-dimension size of $ N $. Similarly, we have 
	\begin{align}
		\bW_{n}\bB\bW_{n}^{\H}= \tilde{\beta}_{n}\bR_{\mathrm{RIS},n}^{1/2}\tr(\bR_{\mathrm{AP},m}\bB)\bR_{\mathrm{RIS},n}^{1/2}\label{lem1_12}.	\end{align}
	
	By inserting \eqref{lem1_11} and \eqref{lem1_12} in \eqref{lem1_1}, we result in
	\begin{align}
		&\EE\{ \bh_{mk}^{\H}\bA\bh_{mk}\bh_{nk}^{\H}\bB\bh_{nk}\} 
		=\tilde{\beta}_{m}\tilde{\beta}_{n}\tr(\bR_{\mathrm{AP},m}\bA)\tr(\bR_{\mathrm{AP},m}\bB)\nn\\
		& \times \EE\{ \bq_{k}^{\H}\bPhi_{w_{k}}^{\H}\bR_{\mathrm{RIS},k}\bPhi_{w_{k}}\bq_{k} \bq_{k}^{\H}\bPhi_{w_{k}}^{\H}\bR_{\mathrm{RIS},k}\bPhi_{w_{k}}\bq_{k}\},\label{lem1_2}\\
		&=\tilde{\beta}_{m}\tilde{\beta}_{n}\tr(\bR_{\mathrm{AP},k}\bA)\tr(\bR_{\mathrm{AP},m}\bB) \EE\{| \bq_{k}^{\H}\bPhi_{w_{k}}^{\H}\bR_{\mathrm{RIS},k}\bPhi_{w_{k}}\bq_{k}|^{2}\},\label{lem1_3}\\
		&=\tilde{\beta}_{m}\tilde{\beta}_{n}\tr(\bR_{\mathrm{AP},k}\bA)\tr(\bR_{\mathrm{AP},m}\bB) \nn\\
		&		\times\EE\{\tilde{ \beta}_{k}^{2}|\bc_{k}^{\H}\bR_{\mathrm{RIS},k}^{1/2}\bPhi_{w_{k}}^{\H}\bR_{\mathrm{RIS},k}\bPhi_{w_{k}}\bR_{\mathrm{RIS},k}^{1/2}\bc_{k}|^{2}\},\label{lem1_4}\\
		&=\tilde{\beta}_{mk}\tilde{\beta}_{nk}\tr(\bR_{\mathrm{AP},m}\bA)\tr(\bR_{\mathrm{AP},m}\bB)  \nn\\
		&\times\Big(\!\!\big(\!\tr(\bR_{\mathrm{RIS},k}\bPhi_{w_{k}}^{\H}\bR_{\mathrm{RIS},k}\bPhi_{w_{k}})\!\big)^{2}\nn\\
		&+\tr(\bR_{\mathrm{RIS},k}\bPhi_{w_{k}}^{\H}\bR_{\mathrm{RIS},k}\bPhi_{w_{k}}\bR_{\mathrm{RIS},k}\bPhi_{w_{k}}^{\H}\bR_{\mathrm{RIS},k}\bPhi_{w_{k}})\Big),\label{lem1_5}\\
		&=\tr(\bR_{\mathrm{AP},m}\bA)\tr(\bR_{\mathrm{AP},m}\bB)  \Big(\!\!\tr(\tilde{\bR}_{mk}\tilde{\bR}_{nk})+\tr(\tilde{\bR}_{mk}\tilde{\bR}_{nk})\!\!\Big),\label{lem1_6}
	\end{align}
	where we have used \eqref{lem1_11} and \eqref{lem1_12} in \eqref{lem1_2}. Moreover, we have applied Lemma \ref{lem0proof1} in \eqref{lem1_5}.
	
	Next, \eqref{exp4} is obtained as
	\begin{align}
		&\EE\{ \bh_{mk}^{\H}\bA\bh_{ml}\bh_{nl}^{\H}\bB\bh_{nk}\}
		\nn\\
		&=\EE\{\bq_{k}^{\H}\bPhi_{w_{k}}^{\H}\bW_{m}\bA\bW_{m}^{\H}\bPhi_{w_{k}}\bq_{l}\bq_{l}^{\H}\bPhi_{w_{l}}^{\H}\bW_{n}\bB\bW_{n}^{\H}\bPhi_{w_{k}}\bq_{k}\}\label{exp41}\\
		&=\tilde{\beta}_{mk}\tilde{\beta}_{nk}\tr(\bR_{\mathrm{AP},m}\bA)\tr(\bR_{\mathrm{AP},m}\bB)\nn\\
		&\times \EE\{\bq_{k}^{\H}\bPhi_{w_{k}}^{\H}\bR_{\mathrm{RIS},m}\bPhi_{w_{l}}\bq_{l}\bq_{l}^{\H}\bPhi_{w_{l}}^{\H}\bR_{\mathrm{RIS},m}\bPhi_{w_{k}}\bq_{k}\}\label{exp53}\\
		&=\tilde{\beta}_{mk}\tilde{\beta}_{nk}\tr(\bR_{\mathrm{AP},m}\bA)\tr(\bR_{\mathrm{AP},m}\bB) \nn\\
		&\times\tr(\bR_{\mathrm{RIS},m}\bPhi_{w_{k}}^{\H}\bR_{\mathrm{RIS},m}\bPhi_{w_{l}}\bq_{l}\bq_{l}^{\H}\bPhi_{w_{l}}^{\H}\bR_{\mathrm{RIS},m}\bPhi_{w_{k}})\nn\\
		&=	\tr(\bR_{\mathrm{AP},m}\bA)\tr(\bR_{\mathrm{AP},m}\bB)	\tr(\tilde{\bR}_{mk}\tilde{\bR}_{nl}),		
	\end{align}
	where, in \eqref{exp53}, we have used \eqref{lem1_11} and \eqref{lem1_12}.
	The expectation in \eqref{exp5} is written as
	\begin{align}
		&\EE\{| \bh_{mk}^{\H}\bA\bh_{mk}|^{2}\}=\EE\{|\bq_{k}^{\H}\bPhi_{w_{k}}^{\H}\bW_{m}\bA\bW_{m}^{\H}\bPhi_{w_{k}}\bq_{k}|^{2}\}\nn\\
		&=\EE\{\bq_{k}^{\H}\bPhi_{w_{k}}^{\H}\bW_{m}\bA\bW_{m}^{\H}\bPhi_{w_{k}}\bq_{k}\bq_{k}^{\H}\bPhi_{w_{k}}^{\H}\bW_{m}\bA^{\H}\bW_{m}^{\H}\bPhi_{w_{k}}\bq_{k}\}\label{exp51}.
	\end{align}
	
	Since $ N $ is asymptotically large, $ \bW_{m}^{\H}\bPhi_{w_{k}}\bq_{k}\bq_{k}^{\H}\bPhi_{w_{k}}^{\H}\bW_{m} $ in \eqref{exp51} can be approximated by the law of large numbers as
	\begin{align}
		&	\bW_{m}^{\H}\bPhi_{w_{k}}\bq_{k}\bq_{k}^{\H}\bPhi_{w_{k}}^{\H}\bW_{m}\nn\\
		&	=
		\tilde{\beta}_{m}\bR_{\mathrm{AP},m}^{1/2}\bD_{m}^{\H}\bR_{\mathrm{RIS},m}^{1/2}\bPhi_{w_{k}}\bq_{k} \bq_{k}^{\H}\bPhi_{w_{k}}^{\H}\bR_{\mathrm{RIS},m}^{1/2}\bD_{m}\bR_{\mathrm{AP},m}^{1/2}\nn\\
		&\approx \tilde{\beta}_{m}\bR_{\mathrm{AP},m}\tr(\bR_{\mathrm{RIS},m}\bPhi_{w_{k}}\bq_{k}\bq_{k}^{\H}\bPhi_{w_{k}}^{\H}). \label{exp52}
	\end{align}
	
	Thus, we insert \eqref{exp52} in \eqref{exp51}, and we  obtain \eqref{exp5} as
	\begin{align}
		&\EE\{| \bh_{mk}^{\H}\bA\bh_{mk}|^{2}\}=\tilde{\beta}_{m}\tr(\bR_{\mathrm{RIS},m}\bPhi_{w_{k}}\bq_{k}\bq_{k}^{\H}\bPhi_{w_{k}}^{\H}) \nn\\
		&	\times\EE\{\bq_{k}^{\H}\bPhi_{w_{k}}^{\H}\bW_{m}\bA\bR_{\mathrm{AP},m}
		\bA^{\H}\bW_{m}^{\H}\bPhi_{w_{k}}\bq_{k}\}\nn\\
		&=\tilde{\beta}_{m}^{2}\tr(\bR_{\mathrm{AP},m}\bA\bR_{\mathrm{AP},m}\bA^{\H})\nn\\
		&		\times\EE\{\bq_{k}^{\H}\bPhi_{w_{k}}^{\H}\bR_{\mathrm{RIS},m}\bPhi_{w_{k}}\bq_{k}\bq_{k}^{\H}\bPhi_{w_{k}}^{\H}\bR_{\mathrm{RIS},m}\bPhi_{w_{k}}\bq_{k}\label{exp52}\}\\
		&=\tilde{\beta}_{m}^{2}\tr(\bR_{\mathrm{AP},m}\bA\bR_{\mathrm{AP},m}\bA^{\H})\EE\{|\bq_{k}^{\H}\bPhi_{w_{k}}^{\H}\bR_{\mathrm{RIS},m}\bPhi_{w_{k}}\bq_{k}|^{2}\}\nn\\
		&=\tilde{\beta}_{mk}^{2}\tr(\bR_{\mathrm{AP},m}\bA\bR_{\mathrm{AP},m}\bA^{\H})\nn\\
		&\times \EE\{|\bc_{k}^{\H}\bR_{\mathrm{RIS},k}^{1/2}\bPhi_{w_{k}}^{\H}\bR_{\mathrm{RIS},k}\bPhi_{w_{k}}\bR_{\mathrm{RIS},k}^{1/2}\bc_{k}|^{2}\}	\nn	\\
		&=\tilde{\beta}_{mk}^{2}\tr(\bR_{\mathrm{AP},m}\bA\bR_{\mathrm{AP},m}\bA^{\H})\nn\\
		&\times\Big(\!\!\tr(\bPhi_{w_{k}}^{\H}\bR_{\mathrm{RIS},k}\bPhi_{w_{k}}\bR_{\mathrm{RIS},k}\bPhi_{w_{k}}^{\H}\bR_{\mathrm{RIS},k})\nn\\
		&+\tr(\bPhi_{w_{k}}^{\H}\bR_{\mathrm{RIS},k}\bPhi_{w_{k}}^{\H}\bR_{\mathrm{RIS},k})\!\!\Big)\label{exp53}\\
		&=\tr(\bR_{\mathrm{AP},m}\bA\bR_{\mathrm{AP},m}\bA^{\H})\Big(\!\!\tr\big(\tilde{\bR}_{mk}^{2}\big)+\big(\!\tr\big(\tilde{\bR}_{mk}\big)\!\big)^{2}\!\Big),
	\end{align}
	where we have used  \eqref{lem1_11} in \eqref{exp52}, and Lemma  \ref{lem0proof1} in \eqref{exp53}.
	
	Finally, the expectation in \eqref{exp6} is written as
	\begin{align}
		&	\EE\{| \bh_{mk}^{\H}\bA\bh_{ml}|^{2}\}=\EE\{|\bq_{k}^{\H}\bPhi_{w_{k}}^{\H}\bW_{m}\bA\bW_{m}^{\H}\bPhi_{w_{l}}\bq_{l}|^{2}\}\nn\\
		&=\EE\{\bq_{k}^{\H}\bPhi_{w_{k}}^{\H}\bW_{m}\bA\bW_{m}^{\H}\bPhi_{w_{l}}\bq_{l}\bq_{l}^{\H}\bPhi_{w_{l}}^{\H}\bW_{m}\bA^{\H}\bW_{m}^{\H}\bPhi_{w_{k}}\bq_{k}\}\nn\\
		&=\tilde{\beta}_{l}\EE\{\bq_{k}^{\H}\bPhi_{w_{k}}^{\H}\bW_{m}\bA\bW_{m}^{\H}\bPhi_{w_{l}}\bR_{\mathrm{RIS},k}
		\bPhi_{w_{l}}^{\H}\bW_{m}\bA^{\H}\bW_{m}^{\H}\bPhi_{w_{k}}\bq_{k}\}\nn\\
		&=\tilde{\beta}_{l}\tilde{\beta}_{k}\EE\{\tr\left(\right.\bR_{\mathrm{RIS},k}\bPhi_{w_{k}}^{\H}\bW_{m}\bA\bW_{m}^{\H}\bPhi_{w_{l}}\bR_{\mathrm{RIS},k}
		\bPhi_{w_{l}}^{\H}\nn\\
		&\times \bW_{m}\bA^{\H}\bW_{m}^{\H}\bPhi_{w_{k}}\left.\right)\}.\label{exp61}
	\end{align}
	
	Using the  fact that $ N $	is   asymptotically large, $\bW_{m}\bA\bW_{m}^{\H}\bPhi_{w_{l}}\bR_{\mathrm{RIS},k}
	\bPhi_{w_{l}}^{\H}\bW_{m}$ can be approximated by the law of large numbers as
	\begin{align}
		&		\bW_{m}\bA\bW_{m}^{\H}\bPhi_{w_{l}}\bR_{\mathrm{RIS},k}
		\bPhi_{w_{l}}^{\H}\bW_{m}\nn\\
		&=\tilde{\beta}_{m}\bR_{\mathrm{AP},m}^{1/2}\bD_{m}^{\H}\bR_{\mathrm{RIS},k}^{1/2}\bPhi_{w_{l}}\bR_{\mathrm{RIS},k}\bPhi_{w_{l}}^{\H}\bR_{\mathrm{RIS},k}^{1/2}\bD_{m}\bR_{\mathrm{AP},m}^{1/2}\nn\\
		&\approx\tilde{\beta}_{m}\bR_{\mathrm{AP},m}\tr(\bPhi_{w_{l}}\bR_{\mathrm{RIS},k}
		\bPhi_{w_{l}}^{\H}\bR_{\mathrm{RIS},k})\nn\\
		&=\frac{1}{\tilde{\beta}_{l}}\bR_{\mathrm{AP},m}\tr(\tilde{\bR}_{ml}).\label{exp62}
	\end{align}
	
	By substituting \eqref{exp62} in \eqref{exp61}, we compute \eqref{exp6} as
	\begin{align}
		&	\EE\{| \bh_{mk}^{\H}\bA\bh_{ml}|^{2}\}\approx\tilde{\beta}_{k}\tr(\tilde{\bR}_{ml})\nn\\
		&\times\EE\{\tr(\bR_{\mathrm{RIS},k}\bPhi_{w_{k}}^{\H}\bW_{m}\bA\bR_{\mathrm{AP},m}
		\bA^{\H}\bW_{m}^{\H}\bPhi_{w_{k}})\}\nn\\
		&= \tilde{\beta}_{mk}\tr(\tilde{\bR}_{ml})\tr(\bR_{\mathrm{RIS},k}
		\bPhi_{w_{k}}^{\H}\bR_{\mathrm{RIS},k}\bPhi_{w_{k}})\tr(\bR_{\mathrm{AP},m}\bA\bR_{\mathrm{AP},m}\bA^{\H})\nn\\
		&=\tr(\tilde{\bR}_{mk})\tr(\tilde{\bR}_{ml})\tr(\bR_{\mathrm{AP},m}\bA\bR_{\mathrm{AP},m}\bA^{\H}).
	\end{align}

	\section{Proof of Lemma~\ref{PropositionDirectChannel}}\label{lem1}
	We follow similar steps to \cite{Kay}. In particular, the LMMSE estimator of $ \bh_{mk} $ is derived by minimizing $ \tr\!\big(\EE\big\{\!(\hat{\bh}_{mk}-{\bh}_{mk})(\hat{\bh}_{mk}-{\bh}_{mk})^{\H}\!\big\}\!\big) $ as
	\begin{align}
		\hat{\bh}_{mk} =\EE\!\left\{\by_{\mathrm{p},mk}\bh_{mk}^{\H}\right\}\left(\EE\!\left\{\by_{\mathrm{p},mk}\by_{\mathrm{p},mk}^{\H}\right\}\right)^{-1}\by_{\mathrm{p},mk}.\label{Cor6}
	\end{align}
	Given that the channel and the receiver noise are uncorrelated, we obtain
	\begin{align}
		\EE\left\{\by_{\mathrm{p},mk}\bh_{mk}^{\H}\right\}
		&=\EE\left\{\bh_{mk}\bh_{mk}^{\H}\right\}\nn\\
		&=\sqrt{p \tau}\bR_{mk}.\label{Cor0}
	\end{align}
	The second term in \eqref{Cor6} is written as
	\begin{align}
		\bQ_{m}=	\EE\left\{\by_{\mathrm{p},mk}\by_{\mathrm{p},mk}^{\H}\right\}={p \tau}\sum_{l\in \mathcal{P}_{k} }\bR_{ml}+\Id_{M}.\label{Cor1}
	\end{align}
	The LMMSE estimate in \eqref{estim1} is derived by substituting  \eqref{Cor0} and \eqref{Cor1} into \eqref{Cor6}. Also,   the covariance matrix of the estimated channel is
	\begin{align}
		\EE\left\{\hat{\bh}_{mk}	\hat{\bh}_{mk}^{\H}\right\}=\bR_{mk}\bQ_{m}\bR_{mk}.\label{var1}
	\end{align} 
	
	\section{Proof of Proposition~\ref{LemmaGradients}}\label{lem2}
	First, we focus on the derivation of  $	\nabla_{\thetv^{t}}\mathrm{NMSE}_{mk}(\thetv,\betv) $. According to \eqref{nmse1}, we have
	\begin{align}
		&		\nabla_{\thetv^{t}}\mathrm{NMSE}_{mk}(\thetv,\betv)
		\nn\\
		&=\frac{\tr(\bPsi_{mk})\nabla_{\boldsymbol{\theta}^{t}}\tr(\bR_{mk})-\tr(\bR_{mk})\nabla_{\boldsymbol{\theta}^{t}}\tr(\bPsi_{mk})}{\tr^{2}(\bR_{mk})}.
	\end{align}
	Based on $ \bR_{mk} $ and $ \bPsi_{mk} $, we observe that $\nabla_{\thetv^{r}}\mathrm{NMSE}_{mk} = 0$, when UE $k$ is in the reflection region, i.e., if $w_k=r$. Hence, we focus on finding $\nabla_{\thetv^{t}}\mathrm{NMSE}_{mk}$ when $w_k=t$. For this reason, we  write $\bR_{mk}$ as
	\begin{align}
		\bR_{mk}&=\tilde{\beta}_{mk}\tr(\bR_{\mathrm{RIS},k} \bPhi_{t} \bR_{\mathrm{RIS},k}  \bPhi_{t}^{\H})\bR_{\mathrm{AP},m}\nn\\ 
		&=\tilde{\beta}_{mk}\tr(\bA_{t}\bPhi_{t}^{\H}  )\bR_{\mathrm{AP},m}\label{cov1t},
	\end{align}
	where $\mathbf{A}_{t}=\mathbf{R}_{\mathrm{RIS},k}\bPhi_{t}\mathbf{R}_{\mathrm{RIS},k}$. In the case  $w_k=r$, we define $\mathbf{A}_{r}=\mathbf{R}_{\mathrm{RIS},k}\bPhi_{r}\mathbf{R}_{\mathrm{RIS},k}$. Next, we focus on the differentials $d(\mathbf{R}_{mk})$ and $d(\boldsymbol{\Psi}_{k})$ that are derived as follows. Regarding $d(\mathbf{R}_{mk})$, we have
	\begin{align}
		&	d(\mathbf{R}_{mk}) 
		=\tilde{\beta}_{mk}\mathbf{R}_{\mathrm{AP}}\tr\bigl(\mathbf{A}_{t}\herm d(\bPhi_{t})+\mathbf{A}_{t}d\bigl(\bPhi_{t}\herm\bigr)\bigr)\nn\\
		&=\tilde{\beta}_{mk}\mathbf{R}_{\mathrm{AP}}\bigl(\bigl(\diag\bigl(\mathbf{A}_{t}\herm\diag(\boldsymbol{{\beta}}^{t})\bigr)\bigr)\trans d(\boldsymbol{\theta}^{t})\nn\\
		&+\bigl(\diag\bigl(\mathbf{A}_{t}\diag(\boldsymbol{{\beta}}^{t})\bigr)\bigr)\trans d(\boldsymbol{\theta}^{t\ast})\bigr)\label{eq:dRk},
	\end{align}
	where \eqref{eq:dRk} is obtained because  $\bPhi_{t}$ is diagonal. Based on \eqref{eq:dRk}, we can conclude that for $w_k=t$, we obtain
	\begin{align}
		\nabla_{\thetv^{t}}\tr(\bR_{mk})&=\frac{\partial}{\partial{\thetv^{t\ast}}}\tr(\bR_{mk})\nn\\
		&=\tilde{\beta}_{mk}\tr(\mathbf{R}_{\mathrm{AP}})\diag\bigl(\mathbf{A}_{t}\diag(\boldsymbol{{\beta}}^{t})\bigr)
	\end{align}
	for $w_k=t$, which indeed proves \eqref{derivtheta}. Similarly,  we can easily obtain $ \nabla_{\thetv^{r}}\tr(\bR_{mk}) $ but we omit the details for brevity.

	Regarding $ d(\boldsymbol{\Psi}_{mk}) $, after application of  \cite[Eq. (3.35)]{hjorungnes:2011}, we obtain
	\begin{align}
		&		d(\boldsymbol{\Psi}_{mk})=p \tau d(\mathbf{R}_{mk}\mathbf{Q}_{k}\mathbf{R}_{mk})\nn\\
		&=p \tau (d(\mathbf{R}_{mk})\mathbf{Q}_{k}\mathbf{R}_{mk}+\mathbf{R}_{mk}d(\mathbf{Q}_{k})\mathbf{R}_{mk}+\mathbf{R}_{mk}\mathbf{Q}_{k}d(\mathbf{R}_{mk})).\label{eq:dPsik}
	\end{align}
	The next step includes the use of \cite[eqn. (3.40)]{hjorungnes:2011}, which gives
	\begin{align}
		d(\mathbf{Q}_{k})  &=d\bigl({p \tau}\sum_{l\in \mathcal{P}_{k} }\bR_{ml}+\Id_{M}\bigr)^{-1}\nn\\
		&=-\bigl({p \tau}\sum_{l\in \mathcal{P}_{k} }\bR_{ml}+\Id_{M}\bigr)^{-1}d\bigl({p \tau}\sum_{l\in \mathcal{P}_{k} }\bR_{ml}+\Id_{M}\bigr)\nn\\
		&\times\bigl({p \tau}\sum_{l\in \mathcal{P}_{k} }\bR_{ml}+\Id_{M}\bigr)^{-1}\nonumber \\
		& =-p \tau \sum_{l\in \mathcal{P}_{k} } \mathbf{Q}_{k}d(\bR_{ml})\mathbf{Q}_{k}.\label{eq:dQk}
	\end{align}
	
	The combination of (\ref{eq:dPsik}) and (\ref{eq:dQk})
	results in
	\begin{align}
		&d(\boldsymbol{\Psi}_{mk})=d(\mathbf{R}_{mk})\mathbf{Q}_{k}\mathbf{R}_{mk}\!-\!p \tau \sum_{l\in \mathcal{P}_{k} } \mathbf{R}_{mk}\mathbf{Q}_{k}d(\bR_{ml})\mathbf{Q}_{k}\mathbf{R}_{mk}\nn\\
		&+\mathbf{R}_{mk}\mathbf{Q}_{k}d\mathbf{R}_{mk}.\label{eq:dPsik-1}
	\end{align}
	
	Thus,  for $w_k=t$, we obtain 
	\begin{align}
		\nabla_{\thetv^{t}}\tr(\bPsi_{mk})&=	\nu_{mk}\diag\bigl(\mathbf{A}_{t}\diag(\boldsymbol{{\beta}}^{t})\bigr)
	\end{align}
	with
	\begin{equation}
		\nu_{mk}\!=\!\hat{\beta}_{mk}\!\tr\bigl(\!\bigl(\mathbf{Q}_{k}\mathbf{R}_{mk}+\mathbf{R}_{mk}\mathbf{Q}_{k}-p \tau \sum_{l\in \mathcal{P}_{k} }\mathbf{Q}_{k}\mathbf{R}_{mk}^{2}\mathbf{Q}_{k}\bigr)\mathbf{R}_{\mathrm{AP}}\bigr).
	\end{equation}
	Similarly,  we can easily obtain $ \nabla_{\thetv^{r}}\tr(\bPsi_{mk}) $.

	For the derivation of  $\nabla_{\betv^{i}}\mathrm{NMSE}_{mk}(\thetv,\betv) $, we have
	\begin{align}
		&		\nabla_{\betv^{i}}\mathrm{NMSE}_{mk}(\thetv,\betv)\nn\\
		&=\frac{\tr(\bPsi_{mk})\nabla_{\boldsymbol{\beta}^{i}}\tr(\bR_{mk})-\tr(\bR_{mk})\nabla_{\boldsymbol{\beta}^{i}}\tr(\bPsi_{mk})}{\tr^{2}(\bR_{mk})}, i=t,r.
	\end{align}

	For $ \nabla_{\boldsymbol{\beta}^{i}}\tr(\bR_{mk}) $,  we can write 
	\begin{subequations}\label{dRk_beta_t}
		\begin{align}
			&			d(\mathbf{R}_{k})  =\hat{\beta}_{k}\mathbf{R}_{\mathrm{BS}}\tr\bigl(\mathbf{A}_{t}\herm d(\bPhi_{t})+\mathbf{A}_{t}d\bigl(\bPhi_{t}\herm\bigr)\bigr)\nn\\
			& =\hat{\beta}_{k}\mathbf{R}_{\mathrm{BS}}\bigl(\diag\bigl(\mathbf{A}_{t}\herm\diag(\btheta^{t})\bigr)^{\T}d(\boldsymbol{\beta}^{t})
			\nn\\
			&+\diag\bigl(\mathbf{A}_{t}\diag(\btheta^{t\ast})\bigr)^{\T}d(\boldsymbol{\beta}^{t})\bigr)\nn\\
			& =2\hat{\beta}_{k}\mathbf{R}_{\mathrm{BS}}\Re\bigl\{\diag\bigl(\mathbf{A}_{t}\herm\diag(\btheta^{t})\bigr\}^{\T} d(\boldsymbol{\beta}^{t}).
		\end{align}
		Hence, we have
		\begin{align}
			\!\!\!\nabla_{\betv^{i}}\tr(\bR_{mk})	=2\tilde{\beta}_{mk}\tr(\mathbf{R}_{\mathrm{AP}})\Re\bigl\{\diag\bigl(\mathbf{A}_{i}\herm\diag(\btheta^{i})\bigr\}.
		\end{align}
		In a similar way, we obtain
		\begin{align}
			\nabla_{\thetv^{i}}\tr(\bPsi_{mk})&=2	\nu_{mk}\diag\bigl(\mathbf{A}_{t}\Re\bigl\{\diag\bigl(\mathbf{A}_{i}\herm\diag(\btheta^{i})\bigr\},
		\end{align}
		which concludes the proof.
	\end{subequations}

	\section{Proof of Proposition~\ref{Proposition:DLSINR}}\label{Proposition1}
	We start with the derivation of the desired signal part by exploiting that the channel estimate and the channel estimation error are uncorrelated. We have
	\begin{align}
		\mathrm{DS}_{k}&=\sqrt{\rho_{\mathrm{d}}}\EE\Big\{\sum_{m=1}^{M}  \sqrt{\eta_{mk}}(\hat{\bh}_{mk}+\bee_{mk})^\H\hat{\bh}_{mk}\Big\}\label{ds1}\\
		&=\sqrt{\rho_{\mathrm{d}}}\sum_{m=1}^{M}  \sqrt{\eta_{mk}}\EE\{\hat{\bh}^\H_{mk}\hat{\bh}_{mk}\}\label{ds2}\\
		&=\sqrt{\rho_{\mathrm{d}}}\sum_{m=1}^{M}  \sqrt{\eta_{mk}}\tr(\bPsi_{mk}),\label{ds4}
	\end{align}
	where \eqref{ds1} results by replacing the channel with the channel estimate and the channel estimation error. Next, \eqref{ds2} takes into account that the channel estimate and the error are uncorrelated. In  \eqref{ds3}, we have used  the property  $\bx^{\H}\by = \tr(\by \bx^{\H})$ for any vectors $\bx$, $\by$.
	
	In the case of the beamforming uncertainty term in the denominator of \eqref{SINR}, we have
	\begin{align}
		&\EE\{|\mathrm{BU}_{k}|^{2}\}=\rho_{\mathrm{d}}\EE\Big\{\Big|\sum_{m=1}^{M}g_{mk}\Big|^{2}\Big\}\nn\\
		&=\rho_{\mathrm{d}}\sum_{m=1}^{M} \sum_{n=1, n\ne m}^{M} \EE\{g_{mk}{g}_{nk}^{*}\}+\rho_{\mathrm{d}}\sum_{m=1}^{M}\EE\{|g_{mk}|^{2}\}\nn\\
		&=\mathcal{I}_{1}+\mathcal{I}_{2},
	\end{align}
	where $ g_{mk}=  \sqrt{\eta_{mk}}{\bh}^\H_{mk}\hat{\bh}_{mk}-\sqrt{\eta_{mk}}\EE\Big\{  {\bh}^\H_{mk}\hat{\bh}_{mk}\Big\} $. Next, $ \mathcal{I}_{1} $ can be written as 
	\begin{align}
		\mathcal{I}_{1}&=\rho_{\mathrm{d}}\sum_{m=1}^{M} \sum_{n=1, n\ne m}^{M} \sqrt{\eta_{mk}} \sqrt{\eta_{nk}} \EE\{{\bh}^\H_{mk}\hat{\bh}_{mk}\hat{\bh}_{nk}^{\H}{\bh}_{nk}\}\nn\\
		&-\rho_{\mathrm{d}} \sqrt{\eta_{mk}} \sqrt{\eta_{nk}}\tr(\bPsi_{mk})\tr(\bPsi_{nk}),\label{I1}
	\end{align}
	where we have used the identities $ \EE\{(X-\EE\{X\})\}\EE\{(Y-\EE\{Y\})\}=\EE\{XY\}-\EE\{X\}\EE\{Y\} $ and $ \EE\{|X-\EE\{X\}|^{2}\}=\EE\{|X|^{2}\}-|\EE\{X\}|^{2} $.
	The first part of $ 	\mathcal{I}_{1} $ is obtained as
	\begin{align}
		&\EE\{{\bh}^\H_{mk}\hat{\bh}_{mk}\hat{\bh}_{nk}^{\H}{\bh}_{nk}\}\nn\\
		&=p \tau \EE\{{\bh}^\H_{mk}\bR_{mk}\bQ_{m} \by_{\mathrm{p},mk}\by_{\mathrm{p},nk}^{\H}\bQ_{n}\bR_{nk} {\bh}_{nk}\}\label{estim2}\\
		&=p \tau \EE\{{\bh}^\H_{mk}\bR_{mk}\bQ_{m}
		(\sqrt{p \tau}\sum_{l\in \mathcal{P}_{k} }\bh_{ml}+\bz_{\mathrm{p},mk}) \nn\\
		&\times(\sqrt{p \tau}\sum_{l\in \mathcal{P}_{k} }\bh_{nl}+\bz_{\mathrm{p},nk})^{\H}\bQ_{n}\bR_{nk} {\bh}_{nk}\}\label{estim3}\\
		&=(p \tau)^{2} \EE\{{\bh}^\H_{mk}\bR_{mk}\bQ_{m}\Big(\sum_{l\in \mathcal{P}_{k} }\bh_{ml}\Big)\Big(\sum_{l\in \mathcal{P}_{k} }\bh_{nl}\Big)^{\H}\bQ_{n}\bR_{nk} {\bh}_{nk}\}\label{estim4}\\
		&=(p \tau)^{2} \EE\{{\bh}^\H_{mk}\bR_{mk}\bQ_{m}\Big(\sum_{l\in \mathcal{P}_{k} }\bh_{ml}\bh_{nl}^{\H}\Big)\bQ_{n}\bR_{nk} {\bh}_{nk}\}\label{estim4}\\
		&=(p \tau)^{2} \EE\{{\bh}^\H_{mk}\bR_{mk}\bQ_{m}\bh_{mk}\bh_{nk}^{\H}\bQ_{n}\bR_{nk} {\bh}_{nk}\}\nn\\
		&+(p \tau)^{2}\sum_{l\in \mathcal{P}_{k}\backslash{k} } \EE\{{\bh}^\H_{mk}\bR_{mk}\bQ_{m}\bh_{mk}\bh_{nk}^{\H}\bQ_{n}\bR_{nk} {\bh}_{nk}\}\label{estim5}\\
		&=(p \tau)^{2} \tr(\bR_{\mathrm{AP},m}\bR_{mk}\bQ_{m})\tr(\bR_{\mathrm{AP},m}\bQ_{n}\bR_{nk})\nn\\
		&\times\big(\tr(\tilde{\bR}_{mk})\tr(\tilde{\bR}_{nk})+\tr(\tilde{\bR}_{mk}\tilde{\bR}_{nk})\big)\nn\\
		&+(p \tau)^{2}\sum_{l\in \mathcal{P}_{k}\backslash{k} } \tr(\bR_{\mathrm{AP},m}\bR_{mk}\bQ_{m})\tr(\bR_{\mathrm{AP},m}\bQ_{n}\bR_{nk})			 \tr(\tilde{\bR}_{mk}\tilde{\bR}_{nl})\label{estim5},
	\end{align}
	where, in \eqref{estim5}, we have used \eqref{exp3} and \eqref{exp4} from Lemma \ref{lem0proof}. By inserting \eqref{estim5} into \eqref{I1}, we obtain $ \mathcal{I}_{1} $.
	
	Moreover, $ \mathcal{I}_{2} $ can be rewritten as
	\begin{align}
		&	\mathcal{I}_{2}=\rho_{\mathrm{d}}\!\!\sum_{m=1}^{M}\eta_{mk}\EE\Big\{ | {\bh}^\H_{mk}\hat{\bh}_{mk}|^{2}\Big\}-\rho_{\mathrm{d}}\!\!\sum_{m=1}^{M}\eta_{mk}|\EE\Big\{ \| \hat{\bh}_{mk}\|^{2}\}|^{2}\nn\\
		&=\!\rho_{\mathrm{d}}\!\!\sum_{m=1}^{M}\!\eta_{mk}\EE\Big\{ | {\bh}^\H_{mk}\hat{\bh}_{mk}|^{2}\Big\}\!-\!\rho_{\mathrm{d}}\!\!\sum_{m=1}^{M}\!\eta_{mk}\tr^{2}(\bPsi_{mk})\label{I2}, 
	\end{align}
	where we have applied the properties  $ \EE\{|X-\EE\{X\}|^{2}\}=\EE\{|X|^{2}\}-|\EE\{X\}|^{2} $ and $ \EE\{|X+Y|^{2}\} =\EE\{|X|^{2}\}+\EE\{|Y|^{2}\}$ in the case of zero-mean uncorrelated random variables.
	Now, we focus on the first part of \eqref{I2}. We substitute \eqref{estim1}, and we obtain
	\begin{align}
		&	\EE\Big\{ | {\bh}^\H_{mk}\hat{\bh}_{mk}|^{2}\Big\}=p \tau	\EE\Big\{ | {\bh}^\H_{mk}\bR_{mk}\bQ_{m} \by_{\mathrm{p},mk}|^{2}\Big\}\\
		&=p \tau \EE\Big\{|{\bh}^\H_{mk}\bR_{mk}\bQ_{m} (\sqrt{p \tau}\sum_{l\in \mathcal{P}_{k}}{\bh}_{ml}+\bz_{\mathrm{p},mk})|^{2}\Big\}\label{I23}\\
		&=p \tau \EE\Big\{|\sqrt{p \tau}{\bh}^\H_{mk}\bR_{mk}\bQ_{m}{\bh}_{mk}\nn\\
		& +
		\sqrt{p \tau}{\bh}^\H_{mk}\bR_{mk}\bQ_{m}\sum_{l\in \mathcal{P}_{k}\backslash{k}}{\bh}_{ml}+{\bh}^\H_{mk}\bR_{mk}\bQ_{m}\bz_{\mathrm{p},mk}|^{2}\Big\}\label{I24}\\
		&=(p \tau)^{2} \tr(\bR_{\mathrm{AP},m}\bR_{mk}\bQ_{m}\bR_{\mathrm{AP},m}\bQ_{m}\bR_{mk})\nn\\
		&\times\Big(\!\!\tr\big(\tilde{\bR}_{mk}^{2}\big)+\big(\!\tr\big(\tilde{\bR}_{mk}\big)\!\big)^{2}\!\Big)\nn\\
		&+	(p \tau)^{2}\sum_{l\in \mathcal{P}_{k}\backslash{k}}\Big(\tr(\tilde{\bR}_{ml})\tr(\bR_{\mathrm{AP},m}\bR_{mk}\bQ_{m}\bR_{\mathrm{AP},m}\bQ_{m}\bR_{mk})\nn\\
		&\times\tr\big(\tilde{\bR}_{mk}\big)\!\!\Big)^{2}\}+	p \tau\EE\{|{\bh}^\H_{mk}\bR_{mk}\bQ_{m}\bz_{\mathrm{p},mk}|^{2}\}\label{I25},
	\end{align} 
	where, in \eqref{I25}, we have used \eqref{exp5} and \eqref{exp6} from Lemma~\ref{lem0proof}. Regarding the last term in \eqref{I25}, we have
	\begin{align}
		&\EE\{|{\bh}^\H_{mk}\bR_{mk}\bQ_{m}\bz_{\mathrm{p},mk}|^{2}\}\nn\\
		&=	\EE\{{\bh}^\H_{mk}\bR_{mk}\bQ_{m}\bz_{\mathrm{p},mk}\bz_{\mathrm{p},mk}^{\H}\bQ_{m}\bR_{mk}{\bh}_{mk}\}\nn\\
		&=\EE\{{\bh}^\H_{mk}\bR_{mk}\bQ_{m}\bQ_{m}\bR_{mk}{\bh}_{mk}\}\nn\\
		&=\tr(\bR_{mk}^{2}\bQ_{m}^{2}),\label{noise1}
	\end{align}
	where we have used \eqref{exp2} of Lemma~\ref{lem0proof}. Hence, we substitute \eqref{noise1} in \eqref{I25}, and we obtain 
	\begin{align}
		&\EE\Big\{ | {\bh}^\H_{mk}\hat{\bh}_{mk}|^{2}\Big\}=(p \tau)^{2} \tr(\bR_{\mathrm{AP},m}\bR_{mk}\bQ_{m}\bR_{\mathrm{AP},m}\bQ_{m}\bR_{mk})\nn\\
		&		\times\Big(\!\!\tr\big(\tilde{\bR}_{mk}^{2}\big)+\big(\!\tr\big(\tilde{\bR}_{mk}\big)\!\big)^{2}\!\Big)+	(p \tau)^{2}\!\!\!\sum_{l\in \mathcal{P}_{k}\backslash{k}}\!\!\!|\tr(\tilde{\bR}_{ml})\nn\\
		&\times\tr(\bR_{\mathrm{AP},m}\bR_{mk}\bQ_{m}\bR_{\mathrm{AP},m}\bQ_{m}\bR_{mk})
		\tr(\tilde{\bR}_{mk})|^{2}
		\nn\\
		&	+	p \tau\tr(\bR_{mk}^{2}\bQ_{m}^{2})\label{I26}.
	\end{align}
	By substituting \eqref{I26} in \eqref{I2}, we obtain $\mathcal{I}_{2}  $. Having obtained $ \mathcal{I}_{1} $ and $ \mathcal{I}_{2} $, we result in  $ \EE\{|\mathrm{BU}_{k}|^{2}\} $.

	Regarding the interference term, we have
	\begin{align}
		&	\!\!\!\EE\{|\mathrm{UI}_{ik}|^{2}\}\!=\! \EE\{|\sum_{m=1}^{M}  \sqrt{\eta_{mi}}{\bh}^\H_{mk}\hat{\bh}_{mi}|^{2}\}\nn\\
		&=p \tau \EE\{|\sum_{m=1}^{M}  \sqrt{\eta_{mi}}{\bh}^\H_{mk}\bR_{mi}\bQ_{m}(\sqrt{p \tau}\sum_{l\in \mathcal{P}_{i}}{\bh}_{ml}+\bz_{\mathrm{p},mi})|^{2}\}\nn\\
		&=(p \tau)^{2} \EE\{|\sum_{m=1}^{M}  \sqrt{\eta_{mi}}{\bh}^\H_{mk}\bR_{mi}\bQ_{m}(\sum_{l\in \mathcal{P}_{i}}{\bh}_{ml})|^{2}\nn\\
		&+p \tau\EE\{||\sum_{m=1}^{M}  \sqrt{\eta_{mi}}{\bh}^\H_{mk}\bR_{mi}\bQ_{m}\bz_{\mathrm{p},mi})|^{2}\}\nn\\
		&=(p \tau)^{2} \EE\{|\sum_{m=1}^{M}  \sqrt{\eta_{mi}}{\bh}^\H_{mk}\bR_{mi}\bQ_{m}(\sum_{l\in \mathcal{P}_{i}}{\bh}_{ml})|^{2}\nn\\
		&+p \tau\sum_{m=1}^{M}  {\eta_{mi}}\tr(\bR_{mi}\bQ_{m}\bQ_{m}\bR_{mi}\bR_{mk})\}\label{int2},
	\end{align}
	where
	\begin{align}
		& \EE\{|\sum_{m=1}^{M}  \sqrt{\eta_{mi}}{\bh}^\H_{mk}\bR_{mi}\bQ_{m}(\sum_{l\in \mathcal{P}_{i}}{\bh}_{ml})|^{2}\nn\\
		&=
		\sum_{m=1}^{M}\sum_{n=1}^{M}\sqrt{\eta_{mi}\eta_{ni}}\nn\\
		&		\times\EE\{{\bh}^\H_{mk}\bR_{mi}\bQ_{m}(\sum_{l\in \mathcal{P}_{i}}{\bh}_{ml})(\sum_{l\in \mathcal{P}_{i}}{\bh}_{ml})^{\H}\bQ_{n}\bR_{ni}{\bh}_{nk}\}\nn\\
		&	=\sum_{m=1}^{M}\sum_{n=1}^{M}\sqrt{\eta_{mi}\eta_{ni}}\sum_{l\in \mathcal{P}_{i}}\EE\{{\bh}^\H_{mk}\bR_{mi}\bQ_{m}{\bh}_{ml}{\bh}_{nl}^{\H}\bQ_{n}\bR_{ni}{\bh}_{nk}\}\nn\\
		&=	\sum_{m=1}^{M}\sqrt{\eta_{mi}\eta_{ni}}\sum_{l\in \mathcal{P}_{i}}\EE\{{\bh}^\H_{mk}\bR_{mi}\bQ_{m}{\bh}_{ml}{\bh}_{ml}^{\H}\bQ_{n}\bR_{ni}{\bh}_{nk}\}\nn\\
		&+	\sum_{m=1}^{M}\sum_{n\ne m}^{M}\sqrt{\eta_{mi}\eta_{ni}}\sum_{l\in \mathcal{P}_{i}}\EE\{{\bh}^\H_{mk}\bR_{mi}\bQ_{m}{\bh}_{ml}{\bh}_{nl}^{\H}\bQ_{n}\bR_{ni}{\bh}_{nk}\}\nn\\
		&=\mathcal{I}_{3}+\mathcal{I}_{4}.
	\end{align}
	Now, we select cases. If $ i\not\in \mathcal{P}_{k} $, then $ \mathcal{P}_{i}\cap\mathcal{P}_{k} =\emptyset$ or $ l\ne k $. In this case, we have
	\begin{align}
		&		\mathcal{I}_{3}=\EE\{|{\bh}^\H_{mk}\bR_{mi}\bQ_{m}{\bh}_{ml}|^{2}\}\label{I31}\\
		&=\tr(\bR_{\mathrm{AP},m}\bR_{mi}\bQ_{m}\bR_{\mathrm{AP},m}\bQ_{m}\bR_{mi})\tr(\tilde{\bR}_{ml})\tr(\tilde{\bR}_{mk}),
	\end{align}
	where we have used \eqref{exp6} of Lemma~\ref{lem0proof}. Also, we have
	\begin{align}
		&	\mathcal{I}_{4}=	\sum_{m=1}^{M}\sum_{n\ne m}^{M}\sqrt{\eta_{mi}\eta_{ni}}\sum_{l\in \mathcal{P}_{i}}\tr(\bR_{\mathrm{AP},m}\bR_{mi}\bQ_{m})\nn\\
		&\times	\tr(\bR_{\mathrm{AP},m}\bQ_{n}\bR_{ni})	 \tr(\tilde{\bR}_{mk}\tilde{\bR}_{nl}),
	\end{align}
	where we have used \eqref{exp4} of Lemma~\ref{lem0proof}.
	
	If $ i\not\in \mathcal{P}_{k}\backslash{k} $, then $ \mathcal{P}_{i}=\mathcal{P}_{k}$. In this case, we obtain
	\begin{align}
		&		\sum_{l\in \mathcal{P}_{i}}	\mathcal{I}_{3}=	\sum_{l\in \mathcal{P}_{k}}	\mathcal{I}_{3}\nn\\
		&=\EE\{|{\bh}^\H_{mk}\bR_{mi}\bQ_{m}{\bh}_{mk}|^{2}\}+
		\sum_{l\in \mathcal{P}_{k}\backslash{k}}	\EE\{|{\bh}^\H_{mk}\bR_{mi}\bQ_{m}{\bh}_{ml}|^{2}\}\nn\\
		&=\tr(\bR_{\mathrm{AP},m}\bR_{mi}\bQ_{m}\bR_{\mathrm{AP},m}\bQ_{m}\bR_{mi})\Big(\!\!\tr\big(\tilde{\bR}_{mk}^{2}\big)+\big(\!\tr\big(\tilde{\bR}_{mk}\big)\!\big)^{2}\!\Big)\nn\\
		&+
		\sum_{l\in \mathcal{P}_{k}\backslash{k}}\tr(\bR_{\mathrm{AP},m}\bR_{mi}\bQ_{m}\bR_{\mathrm{AP},m}\bQ_{m}\bR_{mi})\tr(\tilde{\bR}_{ml})\tr(\tilde{\bR}_{mk}),
	\end{align}
	where we have used \eqref{exp5} and \eqref{exp6}  of Lemma~\ref{lem0proof}.
	
	Also, we have 
	\begin{align}
		&	\sum_{l\in \mathcal{P}_{i}}	\mathcal{I}_{4}=	\sum_{l\in \mathcal{P}_{k}}	\mathcal{I}_{4}\nn\\
		&=\EE\{{\bh}^\H_{mk}\bR_{mi}\bQ_{m}{\bh}_{mk}{\bh}_{nk}^{\H}\bQ_{n}\bR_{ni}{\bh}_{nk}\}\nn\\
		&+\sum_{l\in \mathcal{P}_{k}\backslash{k}}\EE\{{\bh}^\H_{mk}\bR_{mi}\bQ_{m}{\bh}_{ml}{\bh}_{nl}^{\H}\bQ_{n}\bR_{ni}{\bh}^\H_{nk}\}\nn\\
		&=\tr(\bR_{\mathrm{AP},m}\bR_{mi}\bQ_{m})\tr(\bR_{\mathrm{AP},m}\bQ_{n}\bR_{ni})\nn\\
		&\times\big(\tr(\bR_{mk})\tr(\bR_{nk})+\tr(\bR_{mk}\bR_{nk})\big)\nn\\
		&+\sum_{l\in \mathcal{P}_{k}\backslash{k}}\tr(\bR_{\mathrm{AP},m}\bR_{mi}\bQ_{m})\tr(\bR_{\mathrm{AP},m}\bQ_{n}\bR_{ni}) \tr(\tilde{\bR}_{mk}\tilde{\bR}_{nl}),
	\end{align}
	where we have used \eqref{exp3} and \eqref{exp4}  of Lemma~\ref{lem0proof}. Substitution of the corresponding $ \mathcal{I}_{3} $ and $ \mathcal{I}_{4} $ into \eqref{int2} concludes the proof of $ \EE\{|\mathrm{UI}_{ik}|^{2}\} $.

\end{appendices}

\bibliographystyle{IEEEtran}

\bibliography{IEEEabrv,mybib}

\end{document}